\documentclass[12pt,notitlepage]{amsart}%
\usepackage{amssymb}
\usepackage{amsfonts}
\usepackage{graphicx}
\usepackage{amscd}
\usepackage{graphicx}
\usepackage{amsmath}
\usepackage{amsxtra}
\usepackage{footmisc}%
\newtheorem{theorem}{Theorem}
\theoremstyle{plain}

\newtheorem{lemma}{Lemma}
\newtheorem{notation}{Notation}

\newtheorem{proposition}{Proposition}
\newtheorem{remark}{Remark}

\numberwithin{equation}{section}
\numberwithin{theorem}{section}
\numberwithin{lemma}{section}
\numberwithin{proposition}{section}
\numberwithin{corollary}{section}

\ifx\pdfoutput\relax\let\pdfoutput=\undefined\fi
\newcount\msipdfoutput
\ifx\pdfoutput\undefined\else
\ifcase\pdfoutput\else
\ifx\paperwidth\undefined\else
\ifdim\paperheight=0pt\relax\else\pdfpageheight\paperheight\fi
\ifdim\paperwidth=0pt\relax\else\pdfpagewidth\paperwidth\fi
\fi\fi\fi
\begin{document}
\title[Continuous-time Quantum Walks]{$p$-Adic quantum mechanics, infinite potential wells, and continuous-time
quantum walks}
\author[Z\'{u}\~{n}iga-Galindo]{W. A. Z\'{u}\~{n}iga-Galindo}
\address{University of Texas Rio Grande Valley\\
School of Mathematical \& Statistical Sciences\\
One West University Blvd\\
Brownsville, TX 78520, United States}
\email{wilson.zunigagalindo@utrgv.edu}
\author[Mayes]{Nathaniel P. Mayes}
\address{University of Texas Rio Grande Valley\\
School of Mathematical and Statistical Sciences\\
One West University Blvd\\
Brownsville, TX 78520, United States}
\email{nathaniel.mayes01@utrgv.edu}
\thanks{The authors were partially supported by the Debnath Endowed \ Professorship}

\begin{abstract}
This article discusses a $p$-adic version of the infinite potential well in
quantum mechanics (QM). This model describes the confinement of a particle in
a $p$-adic ball. We rigorously solve the Cauchy problem for the
Schr\"{o}dinger equation and determine the stationary solutions. The $p$-adic
balls are fractal objects. By dividing a $p$-adic ball into a finite number of
sub-balls and using the wavefunctions of the infinite potential well, we
construct a continuous-time quantum walk (CTQW) on a fully connected graph,
where each vertex corresponds to a sub-ball in the partition of the original
ball. In this way, we establish a connection between $p$-adic QM and quantum computing.

\end{abstract}
\maketitle

\section{Introduction}

This article continues the investigation of the first author on $p$-adic
quantum mechanics (QM), \cite{Zuniga-AP}-\cite{Zuniga-PhA}, see also
\cite{Zuniga-QM}. This work aims to study well-known quantum-mechanical models
from the perspective of the $p$-adic QM. It focuses on $p$-adic versions of
the infinite potential wells and their physical interpretation. In the
$p$-adic case, a box with infinite walls is described by a function supported
in a $p$-adic ball that is infinite outside the ball. We solve rigorously the
Cauchy problem for the $p$-adic Schr\"{o}dinger equation attached to infinite
wells, and then, using the wavefunctions, we construct some continuous-time
quantum walks (CTQWs)\ on countable fully connected graphs. The CTQWs were
introduced by Farhi and Gutmann, \cite{Farhi-Gutman}. In quantum information
theory, quantum walks are used extensively as algorithmic tools for quantum
computation; see, e.g., \cite{Nielsen-Chuang}. The most prominent examples are
Shor's algorithm \cite{Shor} and Grover's algorithm \cite{Grover}. For a
review of the theoretical and experimental aspects of the CTQWs, the reader
may consult \cite{Mulkne-Blumen}.

In the Dirac-von Neumann formulation of QM, the states of a (closed) quantum
system are vectors of an abstract complex Hilbert space $\mathcal{H}$, and the
observables correspond to linear self-adjoint operators in $\mathcal{H}$,
\cite{Dirac}-\cite{Komech}. A particular choice of space $\mathcal{H}$ goes
beyond the mathematical formulation and belongs to the domain of the physical
practice and intuition, \cite[Chap. 1, Sect. 1.1]{Berezin et al}. In practice,
selecting a particular Hilbert space also implies choosing a topology for the
space. The standard choice $\mathcal{H}=L%
{{}^2}%
(\mathbb{R}^{N})$ implies that we are assuming that space ($\mathbb{R}^{N}$)
is continuous, i.e., it is an arcwise topological space, which means that
there is a continuous path joining any two points in the space. Let us denote
by $\mathbb{Q}_{p}$ the field of $p$-adic numbers; here, $p$ is a fixed prime
number. The space $\mathbb{Q}_{p}^{N}$ is discrete, i.e., the points and the
empty set are the only connected subsets. The Hilbert spaces $L^{2}%
(\mathbb{R}^{N})$, $L^{2}(\mathbb{Q}_{p}^{N})$ are isometric, but the
geometries of the underlying spaces ($\mathbb{R}^{N}$, $\mathbb{Q}_{p}^{N}$)
\ are radically different. By $p$-adic QM, we mean QM on $L^{2}(\mathbb{Q}%
_{p}^{N})$; in this case the time is a real number but the position is a
$p$-adic vector.

In the $p$-adic framework, the time evolution of the state of a quantum system
is controlled by a Schr\"{o}dinger equation, which we assumed is obtained from
a heat equation by performing a temporal Wick rotation, \cite{Zuniga-AP}. The
equation considered here (in natural units) has the form%
\[
i\frac{\partial\Psi(x,t)}{\partial t}=\left(  m_{\alpha}\boldsymbol{D}%
^{\alpha}+V\left(  \left\Vert x\right\Vert _{p}\right)  \right)
\Psi(x,t)\text{, }x\in\mathbb{Q}_{p}^{N},\quad t\geq0,
\]
where $\boldsymbol{D}^{\alpha}$ is the Taibleson-Vladimirov fractional
derivative, $\alpha>0$ and $m_{\alpha}$ is a positive constant;
$\boldsymbol{D}^{\alpha}$ is a nonlocal operator, and $V\left(  \left\Vert
x\right\Vert _{p}\right)  $ is a potential, which is infinite outside of a
ball. In the last forty years, the $p$-adic quantum mechanics and the $p$-adic
Schr\"{o}\-dinger equations have been studied extensively, see, e.g.,
\cite{Beltrameti et al}-\cite{Aniello et al}, among many available references.
There are at least three different types of $p$-adic Schr\"{o}\-dinger
equations. In the first type, the wavefunctions are complex-valued and the
space-time is $\mathbb{Q}_{p}^{N}\times\mathbb{R}$; in the second type, the
wavefunctions are complex-valued and the space-time is $\mathbb{Q}_{p}%
^{N}\times\mathbb{Q}_{p}$; in the third type, the wavefunctions are $p$-adic-
valued and the space-time is $\mathbb{Q}_{p}^{N}\times\mathbb{Q}_{p}$.
Time-independent Schr\"{o}\-dinger equations have been intensively studied, in
particular, the spectra of the corresponding Schr\"{o}\-dinger operators, see,
e.g., \cite[Chap. 2, Sect. X, \ and \ Chap. 3, Sects. \ XI, XII]{V-V-Z},
\cite[Chap. 3]{Kochubei}.

Recently, Aniello \textit{et al}., \cite{p-adic-QC}, proposed a model of a
quantum $N$-dimensional system (quNit) based on a quadratic extension of the
field of $p$-adic numbers. In this approach, each quNit corresponds to a
$p$-adic valued function in a $p$-adic Hilbert space; this fact requires a new
interpretation of the QM in terms of $p$-adic valued probabilities. Our
approach is entirely different. In the case of dimension one, we decompose the
$p$-adic unit ball $\mathbb{Z}_{p}$ as a disjoint union of open compact
subsets $\mathcal{K}_{j}$, $j=1,\ldots,N$; the normalized characteristic
functions of these subsets corresponds with the quNits, which are real-valued
square-integrable functions, i.e. functions in $L^{2}\left(  \mathbb{Z}%
_{p}\right)  $ with norm one. These quNits are a linearly independent set in
$L^{2}\left(  \mathbb{Z}_{p}\right)  $. We compute the transition probability
from $1_{\mathcal{K}_{j}}$ to $1_{\mathcal{K}_{i}}$ by solving a Cauchy
problem associated with the Schr\"{o}dinger equation of the infinite potential
well. This construction opens an interesting connection between $p$-adic QM
and quantum computing.

\section{The field of $p$-adic numbers}

From now on, we use $p$ to denote a fixed prime number; starting in Section
\ref{Section_Eigenvalue_Problems}, we take $p\geq3$. Any non-zero $p-$adic
number $x$ has a unique expansion of the form%
\begin{equation}
x=x_{-k}p^{-k}+x_{-k+1}p^{-k+1}+\ldots+x_{0}+x_{1}p+\ldots,\text{ }
\label{p-adic-number}%
\end{equation}
with $x_{-k}\neq0$, where $k$ is an integer, and the $x_{j}$s\ are numbers
from the set $\left\{  0,1,\ldots,p-1\right\}  $. The set of all possible
sequences of the form (\ref{p-adic-number}) constitutes the field of $p$-adic
numbers $\mathbb{Q}_{p}$. There are natural field operations, sum and
multiplication, on series of form (\ref{p-adic-number}). There is also a norm
in $\mathbb{Q}_{p}$ defined as $\left\vert x\right\vert _{p}=p^{-ord(x)}$,
where $ord(x)=k$, for a nonzero $p$-adic number $x$. By definition
$ord(0)=\infty$. The field of $p$-adic numbers with the distance induced by
$\left\vert \cdot\right\vert _{p}$ is a complete ultrametric space. The
ultrametric property refers to the fact that $\left\vert x-y\right\vert
_{p}\leq\max\left\{  \left\vert x-z\right\vert _{p},\left\vert z-y\right\vert
_{p}\right\}  $ for any $x$, $y$, $z$ in $\mathbb{Q}_{p}$. The $p$-adic
integers which are sequences of form (\ref{p-adic-number}) with $-k\geq0$. All
these sequences constitute the unit ball $\mathbb{Z}_{p}$. The unit ball is an
infinite rooted tree with fractal structure. As a topological space
$\mathbb{Q}_{p}$\ is homeomorphic to a Cantor-like subset of the real line,
see, e.g., \cite{V-V-Z}, \cite{Alberio et al}.

We extend the $p-$adic norm to $\mathbb{Q}_{p}^{N}$ by taking%
\[
||x||_{p}:=\max_{1\leq i\leq N}|x_{i}|_{p},\qquad\text{for }x=(x_{1}%
,\dots,x_{N})\in\mathbb{Q}_{p}^{N}.
\]
We define $ord(x)=\min_{1\leq i\leq N}\{ord(x_{i})\}$, then $||x||_{p}%
=p^{-ord(x)}$.\ The metric space $\left(  \mathbb{Q}_{p}^{N},||\cdot
||_{p}\right)  $ is a complete ultrametric space.

A function $\varphi:\mathbb{Q}_{p}^{N}\rightarrow\mathbb{C}$ is called locally
constant, if for any $a\in\mathbb{Q}_{p}^{N}$, there is an integer $l=l(a)$,
such that
\[
\varphi\left(  a+x\right)  =\varphi\left(  a\right)  \text{ for any }%
||x||_{p}\leq p^{l}.
\]
The set of functions for which $l=l\left(  \varphi\right)  $ depends only on
$\varphi$ form a $\mathbb{C}$-vector space denoted as $\mathcal{U}%
_{loc}\left(  \mathbb{Q}_{p}^{N}\right)  $. We call $l\left(  \varphi\right)
$ the exponent of local constancy. If $\varphi\in\mathcal{U}_{loc}\left(
\mathbb{Q}_{p}^{N}\right)  $ has compact support, we say that $\varphi$ is a
test function. We denote by $\mathcal{D}(\mathbb{Q}_{p}^{N})$ the complex
vector space of test functions. There is a natural integration theory so that
$\int_{\mathbb{Q}_{p}^{N}}\varphi\left(  x\right)  d^{N}x$ gives a
well-defined complex number. The measure $d^{N}x$ is the Haar measure of
$\mathbb{Q}_{p}^{N}$. In the Appendix (Section \ref{Appendix}), we give a
quick review of the basic aspects of the $p$-adic analysis required here.

In this article consider QM in the sense of the Dirac-von Neumann
formulation\ on the Hilbert space%
\[
L^{2}(\mathbb{Q}_{p}^{N}):=L^{2}(\mathbb{Q}_{p}^{N},d^{N}x)=\left\{
f:\mathbb{Q}_{p}^{N}\rightarrow\mathbb{C};\left\Vert f\right\Vert _{2}=\left(
\text{ }%
{\displaystyle\int\limits_{\mathbb{Q}_{p}^{N}}}
\left\vert f\left(  x\right)  \right\vert ^{2}d^{N}x\right)  ^{\frac{1}{2}%
}<\infty\right\}  .
\]
Given $f,g\in L^{2}(\mathbb{Q}_{p}^{N})$, we set
\[
\left\langle f,g\right\rangle =%
{\displaystyle\int\limits_{\mathbb{Q}_{p}^{N}}}
f\left(  x\right)  \overline{g\left(  x\right)  }d^{N}x,
\]
where the bar denotes the complex conjugate.

\section{The Dirac-von Neumann formulation of QM}

In the Dirac-Von Neumann formulation of QM, to every isolated quantum system
there is associated a separable complex Hilbert space $\mathcal{H}$ called the
space of states. The Hilbert space of a composite system is the Hilbert space
tensor product of the state spaces associated with the component systems. The
states of a quantum system are described by non-zero vectors from
$\mathcal{H}$. Two vectors describe the same state if and only if they differ
by a non-zero complex factor. Each observable corresponds to a unique linear
self-adjoint operator in $\mathcal{H}$. \ The most important observable of a
quantum system is its energy. We denote the corresponding operator by
$\boldsymbol{H}$. Let $\Psi_{0}\in\mathcal{H}$ be the state at time $t=0$ of a
certain quantum system. Then at time $t$ the system is represented by the
vector $\Psi\left(  t\right)  =\boldsymbol{U}_{t}\Psi_{0}$, where
$\boldsymbol{U}_{t}=e^{-it\boldsymbol{H}}$, $t\geq0$,\ is a unitary operator
called the evolution operator. The evolution operators $\left\{
\boldsymbol{U}_{t}\right\}  _{t\geq0}$ form a strongly continuous
one-parameter unitary group on $\mathcal{H}$. The vector function $\Psi\left(
t\right)  $ is differentiable if $\Psi\left(  t\right)  $ is contained in the
domain $Dom(\boldsymbol{H})$ of $\boldsymbol{H}$, which happens\ if at $t=0$,
$\Psi_{0}\in Dom(\boldsymbol{H})$, and in this case the time evolution of
$\Psi\left(  t\right)  $ is controlled by the Schr\"{o}dinger equation
\[
i\frac{\partial}{\partial t}\Psi\left(  t\right)  =\boldsymbol{H}\Psi\left(
t\right)  \text{, }%
\]
where $i=\sqrt{-1}$ and the Planck constant is assumed to be one. For an
in-depth discussion of QM, the reader may consult \cite{Berezin et
al}-\cite{Takhtajan}.

In standard QM, the states of quantum systems are functions from spaces of
type $L^{2}(\mathbb{R}^{N})$. The wavefunctions take the form%
\[
\Psi\left(  x,t\right)  :\mathbb{R}^{N}\times\mathbb{R}_{+}\rightarrow
\mathbb{C}\text{,}%
\]
where $x\in\mathbb{R}^{N}$, $t\in\mathbb{R}_{+}=\left\{  t\in\mathbb{R}%
;t\geq0\right\}  $. This choice implies that the space is continuous, i.e.,
given two different points $x_{0}$, $x_{1}\in\mathbb{R}^{N}$ there exists a
continuous curve $X\left(  t\right)  :\left[  a,b\right]  \rightarrow
\mathbb{R}^{N}$ such that $X\left(  a\right)  =x_{0}$, $X\left(  b\right)
=x_{1}$. The Dirac-von Neumann formulation of QM does not rule out the
possibility of choosing a discrete space. In this article we take
$\mathcal{H=}L^{2}(\mathbb{Q}_{p}^{N})$, in this case the wavefunctions take
the form%
\[
\Psi\left(  x,t\right)  :\mathbb{Q}_{p}^{N}\times\mathbb{R}_{+}\rightarrow
\mathbb{C}\text{,}%
\]
where $x\in\mathbb{Q}_{p}^{N}$, $t\in\mathbb{R}_{+}$. The space $\mathbb{Q}%
_{p}^{N}$ is a discrete space. Along this article, discrete space means a
completely disconnected topological space, which is a topological space where
the connected components are the points and the empty set. In such a space, a
continuous curve joining two different points does not exist.

A large class of $p$-adic QM theories are obtained from the Dirac-von Neumann
formulation by taking $\mathcal{H}=L^{2}(\mathbb{Q}_{p}^{N})$. A specific
Schr\"{o}dinger equation controls the dynamic evolution of the quantum states
in each particular QM. In this article, we consider free Schr\"{o}dinger
equations of type
\begin{equation}
i\hbar\frac{\partial\Psi(x,t)}{\partial t}=m_{\alpha}\boldsymbol{D}^{\alpha
}\Psi(x,t),\quad x\in\mathbb{Q}_{p}^{N},\quad t\geq0, \label{Heat_Equation_2}%
\end{equation}
where $m_{\alpha}$ is a positive constant and $\boldsymbol{D}^{\alpha}$ is the
Taibleson-Vladimirov pseudo-differential operator, see Section \ref{Appendix}
for further details.

\section{$p$-Adic Schr\"{o}dinger equations for Infinite well potentials}

The Schr\"{o}dinger equation for a single nonrelativistic particle moving in
$\mathbb{Q}_{p}^{N}$ under a potential
\[
V(x)=V(\left\Vert x\right\Vert _{p})=\left\{
\begin{array}
[c]{ccc}%
0 & \text{if} & x\in p^{L}\mathbb{Z}_{p}^{N}\\
&  & \\
\infty & \text{if} & x\notin p^{L}\mathbb{Z}_{p}^{N},
\end{array}
\right.
\]
is
\begin{equation}
\left\{
\begin{array}
[c]{ll}%
i\hbar\frac{\partial\Psi(x,t)}{\partial t}=m_{\alpha}\boldsymbol{D}^{\alpha
}\Psi(x,t)\text{,} & x\in p^{L}\mathbb{Z}_{p}^{N},\quad t\geq0\\
& \\
\Psi(x,t)=0, & x\notin p^{L}\mathbb{Z}_{p}^{N},\quad t\geq0,
\end{array}
\right.  \label{Schrodinger_Equation_1}%
\end{equation}
where $m_{\alpha}$ is a positive constant. We look for solutions of the
time-dependent Schr\"{o}dinger equation (\ref{Schrodinger_Equation_1}) of the
form
\[
\Psi(x,t)=e^{-\frac{i}{\hbar}Et}\Phi(x),
\]
where $\Phi(x)$ is the time-independent function satisfying
\begin{equation}
\left\{
\begin{array}
[c]{lll}%
\Phi(x)\in L^{2}(p^{L}\mathbb{Z}_{p}^{N}) &  & \\
&  & \\
m_{\alpha}\boldsymbol{D}^{\alpha}\Phi(x)=E\Phi(x)\text{,} & \text{for } & x\in
p^{L}\mathbb{Z}_{p}^{N}\\
&  & \\
\Phi(x)=0 & \text{for } & x\notin p^{L}\mathbb{Z}_{p}^{N}.
\end{array}
\right.  \text{ } \label{Schrodinger_Equation_Ind_0}%
\end{equation}
The boundary of the ball $B_{-L}^{N}=p^{L}\mathbb{Z}_{p}^{N}$ is the empty
set. For this reason in (\ref{Schrodinger_Equation_Ind_0}) there are no
boundary conditions. Since, we are only interested in $\left\vert
\Psi(x,t)\right\vert ^{2}$, we assume without loss of generality that
$\Phi(x)$ is a\ real-valued function.

The Taibleson-Vladimirov pseudo-differential operator $\boldsymbol{D}^{\alpha
}:$ $Dom(\boldsymbol{D}^{\alpha})\rightarrow L^{2}(\mathbb{Q}_{p}^{N})$ is a
non-local operator with domain
\[
Dom(\boldsymbol{D}^{\alpha})=\left\{  f\in L^{2}(\mathbb{Q}_{p}^{N}%
);\left\Vert \xi\right\Vert _{p}^{\alpha}\widehat{f}\in L^{2}(\mathbb{Q}%
_{p}^{N})\right\}  ,
\]
where $\widehat{f}$ is the Fourier transform in $L^{2}(\mathbb{Q}_{p}^{N})$.
In (\ref{Schrodinger_Equation_Ind_0}), the condition $\Phi(x)=0$, for $x\notin
p^{L}\mathbb{Z}_{p}^{N}$ is completely necessary due to the fact
$\boldsymbol{D}^{\alpha}$ is non-local. We denote by $\Omega\left(
p^{L}\left\Vert x\right\Vert _{p}\right)  $ the characteristic function of the
ball $p^{L}\mathbb{Z}_{p}^{N}$. The Hamiltonian%
\[
\boldsymbol{H}:\Phi(x)\rightarrow m_{\alpha}\Omega\left(  p^{L}\left\Vert
x\right\Vert _{p}\right)  \boldsymbol{D}^{\alpha}\Phi(x)
\]
is a well-defined operator on $Dom(\boldsymbol{D}^{\alpha})\cap L^{2}%
(p^{L}\mathbb{Z}_{p}^{N})=\left\{  \Phi\in L^{2}(p^{L}\mathbb{Z}_{p}%
^{N});\left\Vert \xi\right\Vert _{p}^{\alpha}\widehat{\Phi}\in L^{2}%
(\mathbb{Q}_{p}^{N})\right\}  $ into $L^{2}(p^{L}\mathbb{Z}_{p}^{N})$.

We define the inversion operator $\boldsymbol{P}$ as $\boldsymbol{P}%
\psi(x)=\psi(-x)$, for $\psi\in L^{2}(p^{L}\mathbb{Z}_{p}^{N})$. We identify
the potential $V(\left\Vert x\right\Vert _{p})$ with the operator
$\boldsymbol{V}$, which acts as $\boldsymbol{V}\psi(x)=V(\left\Vert
x\right\Vert _{p})\psi(x)=0$, for $\psi\in L^{2}(p^{L}\mathbb{Z}_{p}^{N})$.

Given two operators $\boldsymbol{F}$, $\boldsymbol{G}$, we denote their
commutator by $[\boldsymbol{F},\boldsymbol{G}]=\boldsymbol{FG}-\boldsymbol{GF}%
$.

\begin{lemma}
\label{Lemma0}Take $\psi\in\mathcal{D}\left(  \mathbb{Q}_{p}^{N}\right)  $.
Then, the following assertions hold: (i ) $[\boldsymbol{P},\boldsymbol{D}%
^{\alpha}]\psi=0$; (ii) $[\boldsymbol{P},\boldsymbol{H}]\psi=0$; (iii) if
$\Phi(x)\in\mathcal{D}\left(  \mathbb{Q}_{p}^{N}\right)  $ is a solution of
(\ref{Schrodinger_Equation_Ind_0}), then $\Phi(-x)\in\mathcal{D}\left(
\mathbb{Q}_{p}^{N}\right)  $\ is also a solution; (iv) any solution\ of
(\ref{Schrodinger_Equation_Ind_0}) belonging to $\mathcal{D}\left(
\mathbb{Q}_{p}^{N}\right)  $ has the form $\Psi\left(  x\right)  =\Phi
(x)+\Phi(-x)$, i.e., $\Psi\left(  x\right)  =\Psi\left(  -x\right)  $, where
$\Phi(\pm x)\in\mathcal{D}\left(  \mathbb{Q}_{p}^{N}\right)  $ are solutions
of (\ref{Schrodinger_Equation_Ind_0}).
\end{lemma}

\begin{proof}
(i) The Taibleson-Vladimirov operator $\boldsymbol{D}^{\alpha}$ applied to a
test function $\psi\in\mathcal{D}\left(  \mathbb{Q}_{p}^{N}\right)  $ is given
by
\[
(\boldsymbol{D}^{\alpha}\psi)(x)=\frac{1-p^{\alpha}}{1-p^{-\alpha-N}}%
{\displaystyle\int\limits_{\mathbb{Q}_{p}^{N}}}
\frac{\psi(x-y)-\psi(x)}{\Vert y\Vert_{p}^{\alpha+N}}d^{N}y.
\]
Then
\begin{gather*}
\boldsymbol{D}^{\alpha}(\boldsymbol{P}\psi)(x)=\frac{p^{\alpha}-1}%
{1-p^{-\alpha-N}}%
{\displaystyle\int\limits_{\mathbb{Q}_{p}^{N}}}
\frac{(\boldsymbol{P}\psi)(x)-(\boldsymbol{P}\psi)(y)}{\Vert x-y\Vert
_{p}^{\alpha+N}}d^{N}y=\frac{p^{\alpha}-1}{1-p^{-\alpha-N}}%
{\displaystyle\int\limits_{\mathbb{Q}_{p}^{N}}}
\frac{\psi(-x)-\psi(-y)}{\Vert x-y\Vert_{p}^{\alpha+N}}d^{N}y\\
=\frac{p^{\alpha}-1}{1-p^{-\alpha-N}}%
{\displaystyle\int\limits_{\mathbb{Q}_{p}^{N}}}
\frac{\psi(-x)-\psi(-u-x)}{\Vert u\Vert_{p}^{\alpha+N}}d^{N}u\quad
(\text{taking }u=y-x\text{)}\\
=\frac{1-p^{\alpha}}{1-p^{-\alpha-N}}%
{\displaystyle\int\limits_{\mathbb{Q}_{p}^{N}}}
\frac{\psi(-x-u)-\psi(-x)}{\Vert u\Vert_{p}^{\alpha+N}}d^{N}u=\boldsymbol{PD}%
^{\alpha}\psi(x).
\end{gather*}
(ii) By the first part, for $\psi\in\mathcal{D}\left(  \mathbb{Q}_{p}%
^{N}\right)  $,%
\[
\boldsymbol{HP}\psi(x)=m_{\alpha}\Omega\left(  p^{L}\left\Vert x\right\Vert
_{p}\right)  \boldsymbol{D}^{\alpha}\boldsymbol{P}\psi(x)=m_{\alpha}%
\Omega\left(  p^{L}\left\Vert x\right\Vert _{p}\right)  \boldsymbol{PD}%
^{\alpha}\psi(x).
\]
On the other hand,
\begin{gather*}
\boldsymbol{PH}\psi(x)=\boldsymbol{P}\left(  m_{\alpha}\Omega\left(
p^{L}\left\Vert x\right\Vert _{p}\right)  \boldsymbol{D}^{\alpha}%
\psi(x)\right) \\
=m_{\alpha}\Omega\left(  p^{L}\left\Vert -x\right\Vert _{p}\right)
\boldsymbol{PD}^{\alpha}\psi(x)=m_{\alpha}\Omega\left(  p^{L}\left\Vert
x\right\Vert _{p}\right)  \boldsymbol{PD}^{\alpha}\psi(x).
\end{gather*}
Consequently, $[\boldsymbol{P},\boldsymbol{H}]\psi(x)=0$.

(iii)-(iv) Finally, if $\boldsymbol{H}\Phi=E\Phi$, for $\Phi(x)\in
\mathcal{D}\left(  \mathbb{Q}_{p}^{N}\right)  $, by (ii), $\boldsymbol{HP}%
\Phi=\boldsymbol{PH}\Phi=\boldsymbol{P}E\Phi=E\boldsymbol{P}\Phi$.
\end{proof}

\begin{remark}
Take $\psi\in\mathcal{D}\left(  p^{L}\mathbb{Z}_{p}^{N}\right)  $, then%
\[
\Omega\left(  p^{L}\left\Vert x\right\Vert _{p}\right)  \boldsymbol{D}%
^{\alpha}\psi\left(  x\right)  =\frac{1-p^{\alpha}}{1-p^{-\alpha-N}}%
{\displaystyle\int\limits_{p^{L}\mathbb{Z}_{p}^{N}}}
\frac{\psi(x)-\psi(y)}{\Vert x-y\Vert_{p}^{\alpha+N}}d^{N}y+\frac{1-p^{\alpha
}}{1-p^{-\alpha-N}}\lambda_{0}\psi(x),
\]
where
\[
\lambda_{0}=%
{\displaystyle\int\limits_{\mathbb{Q}_{p}^{N}\smallsetminus p^{L}%
\mathbb{Z}_{p}^{N}}}
\frac{d^{N}y}{\Vert y\Vert_{p}^{\alpha+N}}.
\]
This formula shows the non-local behavior of $\Omega\left(  p^{L}\left\Vert
x\right\Vert _{p}\right)  \boldsymbol{D}^{\alpha}$.
\end{remark}

\section{An \textbf{Orthonormal basis for }$L^{2}(p^{L}\mathbb{Z}_{p}^{N})$}

\subsubsection{Additive characters}

We recall that a $p$-adic number $x\neq0$ has a unique expansion of the form
$x=p^{ord(x)}\sum_{j=0}^{\infty}x_{j}p^{j},$ where $x_{j}\in\{0,\dots,p-1\}$
and $x_{0}\neq0$. By using this expansion, we define the fractional part
of\textit{ }$x\in\mathbb{Q}_{p}$, denoted $\{x\}_{p}$, as the rational number
\[
\left\{  x\right\}  _{p}=\left\{
\begin{array}
[c]{lll}%
0 & \text{if} & x=0\text{ or }ord(x)\geq0\\
&  & \\
p^{ord(x)}\sum_{j=0}^{-ord(x)-1}x_{j}p^{j} & \text{if} & ord(x)<0.
\end{array}
\right.
\]
Set $\chi_{p}(y)=\exp(2\pi i\{y\}_{p})$ for $y\in\mathbb{Q}_{p}$. The map
$\chi_{p}(\cdot)$ is an additive character on $\mathbb{Q}_{p}$, i.e., a
continuous map from $\left(  \mathbb{Q}_{p},+\right)  $ into $S$ (the unit
circle considered as a multiplicative group) satisfying $\chi_{p}(x_{0}%
+x_{1})=\chi_{p}(x_{0})\chi_{p}(x_{1})$, $x_{0},x_{1}\in\mathbb{Q}_{p}$. \ The
additive characters of $\mathbb{Q}_{p}$ form an Abelian group which is
isomorphic to $\left(  \mathbb{Q}_{p},+\right)  $. The isomorphism is given by
$\xi\rightarrow\chi_{p}(\xi x)$, see, e.g., \cite[Section 2.3]{Alberio et al}.

\subsubsection{\label{Section p-Adic Wavelets}$p$-Adic wavelets}

Given $\xi=\left(  \xi_{1},\ldots,\xi_{N}\right)  $, $x=\left(  x_{1}%
,\ldots,x_{N}\right)  \in\mathbb{Q}_{p}^{N}$, we set $\xi\cdot x=\sum
_{i=1}^{N}\xi_{i}x_{i}$. We denote by $\Omega\left(  p^{L}\left\Vert
x-a\right\Vert _{p}\right)  $ the characteristic function of the ball
\[
B_{-L}^{N}(a)=\left\{  x\in\mathbb{Q}_{p}^{N};\left\Vert x-a\right\Vert
_{p}\leq p^{-L}\right\}  =a+p^{L}\mathbb{Z}_{p}^{N}.
\]
We now define%
\[
\mathbb{Q}_{p}/\mathbb{Z}_{p}=\left\{  \sum_{j=-1}^{-m}x_{j}p^{j};\text{for
some }m>0\right\}  .
\]
For $b=\left(  b_{1},\ldots,b_{N}\right)  \in\left(  \mathbb{Q}_{p}%
/\mathbb{Z}_{p}\right)  ^{N}$, $r\in\mathbb{Z}$, we denote by $\Omega\left(
\left\Vert p^{r}x-b\right\Vert _{p}\right)  $ the characteristic function of
the ball $bp^{-r}+p^{-r}\mathbb{Z}_{p}^{N}$.

Set
\begin{equation}
\Psi_{rbk}\left(  x\right)  =p^{\frac{-rN}{2}}\chi_{p}(p^{-1}k\cdot\left(
p^{r}x-b\right)  )\Omega\left(  \left\Vert p^{r}x-b\right\Vert _{p}\right)  ,
\label{Basis_0}%
\end{equation}
where $r\in\mathbb{Z}$, $k=\left(  k_{1},\ldots,k_{N}\right)  \in
\{0,\dots,p-1\}^{N}$, $k\neq\left(  0,\ldots,0\right)  $, and $b=\left(
b_{1},\ldots,b_{N}\right)  \in\left(  \mathbb{Q}_{p}/\mathbb{Z}_{p}\right)
^{N}$. With this notation,%
\begin{equation}
\boldsymbol{D}^{\alpha}\Psi_{rbk}\left(  x\right)  =p^{\left(  1-r\right)
\alpha}\Psi_{rbk}\left(  x\right)  \text{, for any }r,b,k. \label{Basis}%
\end{equation}
Moreover,
\begin{equation}%
{\displaystyle\int\limits_{\mathbb{Q}_{p}^{N}}}
\Psi_{rbk}\left(  x\right)  d^{N}x=0, \label{Average}%
\end{equation}
and $\left\{  \Psi_{rbk}\left(  x\right)  \right\}  _{rbk}$ forms a complete
orthonormal basis of $L^{2}(\mathbb{Q}_{p}^{N})$, see, e.g., \cite[Theorems
9.4.5 and 8.9.3]{Alberio et al}, or \cite[Theorem 3.3]{KKZuniga}.

The restriction of $\Psi_{rbk}\left(  x\right)  $\ to the ball $p^{L}%
\mathbb{Z}_{p}^{N}$ has the form
\begin{equation}
\Omega\left(  p^{L}\left\Vert x\right\Vert _{p}\right)  \Psi_{rbk}\left(
x\right)  =\left\{
\begin{array}
[c]{lll}%
\Psi_{rbk}\left(  x\right)  & \text{if} & bp^{-r}\in p^{L}\mathbb{Z}_{p}%
^{N}\text{, \ }r\leq-L\\
&  & \\
p^{-\frac{rN}{2}}\Omega\left(  p^{L}\left\Vert x\right\Vert _{p}\right)  &
\text{if} & bp^{-r}\in p^{-r}\mathbb{Z}_{p}^{N}\text{, \ }r\geq-L+1\\
&  & \\
0 & \text{if} & bp^{-r}\notin p^{-r}\mathbb{Z}_{p}^{N}\text{, \ }r\geq-L+1,
\end{array}
\right.  \label{Restriction}%
\end{equation}
cf. \cite[Lemma 10.1]{Zuniga-PhA}.

The set of functions%
\[%
{\displaystyle\bigcup}
\left\{  p^{\frac{LN}{2}}\Omega\left(  p^{L}\left\Vert x\right\Vert
_{p}\right)  \right\}  \text{ \ \ }%
{\displaystyle\bigcup}
\text{ \ }%
{\displaystyle\bigcup\limits_{k\in\{0,\dots,p-1\}^{N}\smallsetminus\left\{
0\right\}  }}
\text{ \ \ }%
{\displaystyle\bigcup\limits_{r\leq-L}}
\text{ \ \ }%
{\displaystyle\bigcup\limits_{\substack{bp^{-r}\in p^{L}\mathbb{Z}_{p}%
^{N}\\b\in\left(  \mathbb{Q}_{p}/\mathbb{Z}_{p}\right)  ^{N}}}}
\left\{  \Psi_{rbk}\right\}
\]
is an orthonormal basis in $L^{2}(p^{L}\mathbb{Z}_{p}^{N})$, \cite[Proposition
2]{Zuniga-PhysicaA}, i.e., for $\Phi\left(  x\right)  \in L^{2}(p^{L}%
\mathbb{Z}_{p}^{N})$,
\begin{equation}
\Phi\left(  x\right)  =C_{0}p^{\frac{LN}{2}}\Omega\left(  p^{L}\left\Vert
x\right\Vert _{p}\right)  +%
{\displaystyle\sum\limits_{k}}
\text{ }%
{\displaystyle\sum\limits_{r\leq-L}}
\text{ \ }%
{\displaystyle\sum\limits_{b\in p^{L+r}\mathbb{Z}_{p}^{N}}}
C_{rbk}\Psi_{rbk}\left(  x\right)  , \label{General_form}%
\end{equation}
where the $C_{0}$, $C_{rbk}\in\mathbb{C}$.

\begin{lemma}
\label{Lemma2}Take $p\geq3$, and $\Psi_{rbk}\left(  x\right)  $ as before. If
\begin{equation}
\Psi_{rbk}(x)=\Psi_{rbk}(-x),\text{ for every }x\in bp^{-r}+p^{-r}%
\mathbb{Z}_{p}^{N}, \label{Condition}%
\end{equation}
then $b=0$.
\end{lemma}

\begin{proof}
The formula (\ref{Condition}) implies that
\[
\text{\textrm{supp}}\,\Psi_{rbk}(x)=bp^{-r}+p^{-r}\mathbb{Z}_{p}^{N}%
=-bp^{-r}+p^{-r}\mathbb{Z}_{p}^{N}=\text{\textrm{supp}}\,\Psi_{rbk}(-x).
\]
The point $bp^{-r}\in bp^{-r}+p^{-r}\mathbb{Z}_{p}^{N}$ belong also to
$-bp^{-r}+p^{-r}\mathbb{Z}_{p}^{N}$, consequently,%
\begin{align}
\left\Vert bp^{-r}-(-bp^{-r})\right\Vert _{p}  &  \leq p^{r}\nonumber\\
\left\Vert 2bp^{-r}\right\Vert _{p}  &  \leq p^{r}\nonumber\\
|2|_{p}\left\Vert b\right\Vert _{p}|p^{-r}|_{p}  &  \leq p^{r}\nonumber\\
|2|_{p}\left\Vert b\right\Vert _{p}  &  \leq1. \label{eq:inequality_oneA}%
\end{align}
By using that
\[
|2|_{p}=%
\begin{cases}
2^{-ord_{p}(2)}=\frac{1}{2} & \text{ if }p=2,\\
p^{-ord_{p}(2)}=1 & \text{ if }p\geq3,
\end{cases}
\]
and $p\geq3$, it follows from (\ref{eq:inequality_oneA}) that $\left\Vert
b\right\Vert _{p}\leq1$, with $b\in\left(  \mathbb{Q}_{p}/\mathbb{Z}%
_{p}\right)  ^{N}$. However, if $b\neq0$, then
\[
\left\Vert b\right\Vert _{p}=\max_{1\leq i\leq N}|b_{i}|_{p}=|b_{i_{0}}%
|_{p}>0,\text{ where }b_{i_{0}}\in\left(  \mathbb{Q}_{p}/\mathbb{Z}%
_{p}\right)  \smallsetminus\left\{  0\right\}  ,
\]

implying that there exists a positive integer $m$ such that
\[
|b_{i_{0}}|_{p}=|b_{-m}^{\left(  i_{0}\right)  }p^{-m}+\cdots+b_{-1}^{\left(
i_{0}\right)  }p^{-1}|_{p}=p^{m}>1.
\]
Therefore, in order to satisfy (\ref{eq:inequality_oneA}) $b$ must be zero.
\end{proof}

\begin{remark}
From now on, we assume that $p\geq3$.
\end{remark}

\section{\label{Section_Eigenvalue_Problems}Some Eigenvalue problems}

\begin{lemma}
\label{Lemma3}The following formulae hold: (i)%
\[
\boldsymbol{D}^{\alpha}\Omega\left(  p^{L}\left\Vert x\right\Vert _{p}\right)
=\left\{
\begin{array}
[c]{lll}%
\frac{\left(  1-p^{-N}\right)  p^{-L\alpha}}{\left(  1-p^{-\alpha-N}\right)  }
& \text{if} & x\in p^{L}\mathbb{Z}_{p}^{N}\\
&  & \\
p^{-LN}\left(  \frac{1-p^{-N}}{1-p^{-N-\alpha}}\right)  \frac{1}{\left\Vert
x\right\Vert _{p}^{N+\alpha}} & \text{if} & x\notin p^{L}\mathbb{Z}_{p}^{N}.
\end{array}
\right.
\]

(ii) Set $E_{gnd}:=m_{\alpha}\frac{\left(  1-p^{-N}\right)  p^{-L\alpha}%
}{\left(  1-p^{-\alpha-N}\right)  }$. Then any solution of the eigenvalue
problem
\[
\left\{
\begin{array}
[c]{l}%
\Phi_{E_{gnd}}(x)\in L_{\mathbb{R}}^{2}(p^{L}\mathbb{Z}_{p}^{N})\\
\\
m_{\alpha}\boldsymbol{D}^{\alpha}\Phi_{E_{gnd}}(x)=E_{gnd}\Phi_{E_{gnd}}(x)
\end{array}
\right.
\]
has the form $\Phi_{E_{gnd}}(x)=C_{0}p^{\frac{LN}{2}}\Omega\left(
p^{L}\left\Vert x\right\Vert _{p}\right)  $, where $C_{0}\in\mathbb{R}$.
\end{lemma}

\begin{proof}
(i) We first notice that
\begin{gather*}
\boldsymbol{D}^{\alpha}\Omega\left(  p^{L}\left\Vert x\right\Vert _{p}\right)
=\frac{1-p^{\alpha}}{1-p^{-\alpha-N}}%
{\displaystyle\int\limits_{p^{L}\mathbb{Z}_{p}^{N}}}
\frac{\Omega(p^{L}\left\Vert z\right\Vert _{p})-\Omega\left(  p^{L}\left\Vert
x\right\Vert _{p}\right)  }{\Vert x-z\Vert_{p}^{\alpha+N}}d^{N}z\\
+\frac{1-p^{\alpha}}{1-p^{-\alpha-N}}%
{\displaystyle\int\limits_{\mathbb{Q}_{p}^{N}\smallsetminus p^{L}%
\mathbb{Z}_{p}^{N}}}
\frac{\Omega(p^{L}\left\Vert z\right\Vert _{p})-\Omega\left(  p^{L}\left\Vert
x\right\Vert _{p}\right)  }{\Vert x-z\Vert_{p}^{\alpha+N}}d^{N}z.
\end{gather*}
Now, if $x\in p^{L}\mathbb{Z}_{p}^{N}$, by using that $\left\Vert
x-z\right\Vert _{p}=\left\Vert z\right\Vert _{p}$ for $z\in\mathbb{Q}_{p}%
^{N}\smallsetminus p^{L}\mathbb{Z}_{p}^{N}$,%
\begin{gather*}
\boldsymbol{D}^{\alpha}\Omega\left(  p^{L}\left\Vert x\right\Vert _{p}\right)
=-\frac{1-p^{\alpha}}{1-p^{-\alpha-N}}%
{\displaystyle\int\limits_{\mathbb{Q}_{p}^{N}\smallsetminus p^{L}%
\mathbb{Z}_{p}^{N}}}
\frac{1}{\Vert z\Vert_{p}^{\alpha+N}}d^{N}z\\
=-\frac{1-p^{\alpha}}{1-p^{-\alpha-N}}%
{\displaystyle\sum\limits_{j=L+1}^{\infty}}
\text{ }%
{\displaystyle\int\limits_{\Vert z\Vert_{p}=p^{j}}}
\frac{1}{\Vert z\Vert_{p}^{\alpha+N}}d^{N}z\\
=-\frac{\left(  1-p^{\alpha}\right)  \left(  1-p^{-N}\right)  }{1-p^{-\alpha
-N}}p^{-(L+1)\alpha}%
{\displaystyle\sum\limits_{j=0}^{\infty}}
p^{-j\alpha}=-\frac{\left(  1-p^{\alpha}\right)  \left(  1-p^{-N}\right)
p^{-(L+1)\alpha}}{\left(  1-p^{-\alpha-N}\right)  \left(  1-p^{-\alpha
}\right)  }\\
=\frac{\left(  1-p^{-N}\right)  p^{-L\alpha}}{\left(  1-p^{-\alpha-N}\right)
}.
\end{gather*}

Now, if $x\notin p^{L}\mathbb{Z}_{p}^{N}$, by using that $\left\Vert
x-z\right\Vert _{p}=\left\Vert x\right\Vert _{p}$ for $z\in p^{L}%
\mathbb{Z}_{p}^{N}$,%

\[
\boldsymbol{D}^{\alpha}\Omega\left(  p^{L}\left\Vert x\right\Vert _{p}\right)
=\frac{1-p^{\alpha}}{1-p^{-\alpha-N}}%
{\displaystyle\int\limits_{p^{L}\mathbb{Z}_{p}^{N}}}
\frac{1}{\Vert x\Vert_{p}^{\alpha+N}}d^{N}z=\frac{1-p^{\alpha}}{1-p^{-\alpha
-N}}p^{-LN}\frac{1}{\Vert x\Vert_{p}^{\alpha+N}}.
\]
(ii) It is sufficient to show the uniqueness of the eigenvalue problem in
$L^{2}(p^{L}\mathbb{Z}_{p}^{N})$. By (i), all the functions of the form
$C_{0}p^{\frac{LN}{2}}\Omega\left(  p^{L}\left\Vert x\right\Vert _{p}\right)
$, with $C_{0}\in\mathbb{C}$, are solutions of the eigenvalue problem. Suppose
that
\[
\Phi_{E_{gnd}}(x)=C_{0}p^{\frac{LN}{2}}\Omega\left(  p^{L}\left\Vert
x\right\Vert _{p}\right)  +\Theta\left(  x\right)  \text{, with }\Theta\in
L^{2}(p^{L}\mathbb{Z}_{p}^{N})\text{,}%
\]
is a solution to the eigenvalue problem. Now,$\ \Theta\left(  x\right)
\notin\mathbb{C}\Psi_{rbk}\left(  x\right)  $, otherwise $C_{rbk}\Psi
_{rbk}\left(  x\right)  $ would be a solution of the eigenvalue problem, which
is not possible due to (\ref{Basis}). Then,%
\[
\Theta\left(  x\right)  \notin%
{\displaystyle\bigoplus\limits_{rbk}}
\mathbb{C}\Psi_{rbk}\left(  x\right)  ,
\]
and since
\[
\Theta\left(  x\right)  \in\mathbb{C}\Omega\left(  p^{L}\left\Vert
x\right\Vert _{p}\right)
{\displaystyle\bigoplus}
{\displaystyle\bigoplus\limits_{rbk}}
\mathbb{C}\Psi_{rbk}\left(  x\right)  ,
\]
necessarily $\Theta\left(  x\right)  \in$ $\mathbb{C}\Omega\left(
p^{L}\left\Vert x\right\Vert _{p}\right)  $.
\end{proof}

We denote by $L_{\mathbb{R}}^{2}(p^{L}\mathbb{Z}_{p}^{N})$ the real space of
square-integrable functions defined on the ball $p^{L}\mathbb{Z}_{p}^{N}$.

\begin{remark}
\label{Nota_3}We set $\mathbb{F}_{p}:=\{0,1,\ldots,p-1\}$. We define on
$\mathbb{F}_{p}\times$ $\mathbb{F}_{p}$ the sum modulo $p$, and thus $\left(
\mathbb{F}_{p},+\right)  $ becomes an additive group. If $p\geq3$, then
$-k=p-k$, for $k\in\mathbb{F}_{p}\smallsetminus\left\{  0\right\}  $. We now
consider $\mathbb{F}_{p}^{N}=\{0,1,\ldots,p-1\}^{N}$ as an additive group by
taking the sum modulo $p$ componentwise. In this case, if $k=\left(
k_{1},\ldots,k_{N}\right)  \in\mathbb{F}_{p}^{N}\smallsetminus\left\{
0\right\}  $, then
\[
-k=-\left(  k_{1},\ldots,k_{N}\right)  =\left(  p-k_{1},\ldots,p-k_{N}\right)
=:\boldsymbol{1}p-k.
\]

On the other hand, $\overline{\Psi_{r0k}(x)}=\Psi_{r0k}(-x)=\Psi_{r0\left(
\boldsymbol{1}p-k\right)  }(x)$. Indeed, by using that $\chi_{p}(p^{r}%
x_{j})=1$ for $x_{j}\in p^{-r}\mathbb{Z}_{p}$, $j=1,\ldots,N$,
\begin{gather*}
\Psi_{r0\left(  \boldsymbol{1}p-k\right)  }(x)=p^{\frac{-rN}{2}}\chi
_{p}\left(  p^{r-1}\left(  \boldsymbol{1}p-k\right)  \cdot x\right)
\Omega\left(  \Vert p^{r}x\Vert_{p}\right) \\
=p^{\frac{-rN}{2}}\Omega\left(  \Vert p^{r}x\Vert_{p}\right)  \prod_{j=1}%
^{N}\chi_{p}(p^{r-1}(p-k_{j})x_{j})\\
=p^{\frac{-rN}{2}}\Omega\left(  \Vert p^{r}x\Vert_{p}\right)  \prod_{j=1}%
^{N}\chi_{p}(p^{r}x_{j})\prod_{j=1}^{N}\chi_{p}(-p^{r-1}k_{j}x_{j}%
)=p^{\frac{-rN}{2}}\Omega\left(  \Vert p^{r}x\Vert_{p}\right)  \prod_{j=1}%
^{N}\chi_{p}(-p^{r-1}k_{j}x_{j})\\
=p^{\frac{-rN}{2}}\Omega\left(  \Vert p^{r}x\Vert_{p}\right)  \chi_{p}\left(
p^{r-1}k\cdot\left(  -x\right)  \right)  =\Psi_{r0k}(-x).
\end{gather*}

\end{remark}

\begin{lemma}
\label{Lemma4} Set $E_{r}=m_{\alpha}p^{(1-r)\alpha}$, for a fixed $r\leq-L$.
Then, any solution of eigenvalue problem%
\begin{equation}
\left\{
\begin{array}
[c]{l}%
\Phi_{E_{r}}(x)\in L_{\mathbb{R}}^{2}(p^{L}\mathbb{Z}_{p}^{N})\\
\\
m_{\alpha}\boldsymbol{D}^{\alpha}\Phi_{E_{r}}(x)=E_{r}\Phi_{E_{r}}(x)
\end{array}
\right.  \label{Eigenvaue_problem}%
\end{equation}
is a real-valued function having the form
\[
\Phi_{E_{r}}(x)=2p^{\frac{-rN}{2}}\Omega(\Vert p^{r}x\Vert_{p})\sum
_{k\in\mathbb{F}_{p}^{N}\smallsetminus\left\{  0\right\}  }C_{k}\cos{\left(
2\pi\{p^{r-1}k\cdot x\}_{p}\right)  },
\]
where the $C_{k}\in\mathbb{R}$.
\end{lemma}

\begin{proof}
Any function $\Psi_{rbk}(x)$ with $b\in p^{L+r}\mathbb{Z}_{p}^{N}$, $r\leq-L$,
$k\in\{0,1,\ldots,p-1\}^{N}\smallsetminus\left\{  0\right\}  $, $b\in\left(
\mathbb{Q}_{p}/\mathbb{Z}_{p}\right)  ^{N}$ satisfies
\[
\left\{
\begin{array}
[c]{l}%
\Psi_{rbk}\in L^{2}(p^{L}\mathbb{Z}_{p}^{N})\\
\\
m_{\alpha}\boldsymbol{D}^{\alpha}\Psi_{rbk}(x)=E_{r}\Psi_{rbk}(x).
\end{array}
\right.
\]
Therefore, by (\ref{General_form}), the solution to (\ref{Eigenvaue_problem})
has the form
\[
\Phi_{E_{r}}^{\prime}(x)=\sum_{bk}C_{bk}\Psi_{rbk}(x),
\]
where the $C_{bk}\in\mathbb{C}$. In particular, by taking $\Phi_{E_{r}%
}^{\prime}(x)=\Psi_{rbk}(x)$, and using Lemma \ref{Lemma0}-(iv), $\Psi
_{rbk}(x)=\Psi_{rbk}(-x)$, which implies that $\Psi_{r0k}(x)$ is a solution to
(\ref{Eigenvaue_problem}), cf. Lemma \ref{Lemma2}. Then,
\begin{equation}
\Phi_{E_{r}}^{\prime}(x)=\sum_{k\in\mathbb{F}_{p}^{N}\smallsetminus\left\{
0\right\}  }C_{0k}\Psi_{r0k}(x)=\sum_{k\in\mathbb{F}_{p}^{N}\smallsetminus
\left\{  0\right\}  }C_{k}\Psi_{r0k}(x). \label{Expansion}%
\end{equation}
Again, by Lemma \ref{Lemma0}-(iv), $\Psi_{r0k}(x)+\Psi_{r0k}(-x)$ is a
solution to (\ref{Eigenvaue_problem}), therefore the function%
\begin{align*}
\Phi_{E_{r}}(x)  &  =\sum_{k\in\mathbb{F}_{p}^{N}\smallsetminus\left\{
0\right\}  }C_{k}\left(  \Psi_{r0k}(x)+\Psi_{r0k}(-x)\right)  =\sum
_{k\in\mathbb{F}_{p}^{N}\smallsetminus\left\{  0\right\}  }C_{k}\left(
\Psi_{r0k}(x)+\overline{\Psi_{r0k}(x)}\right) \\
&  =\sum_{k\in\mathbb{F}_{p}^{N}\smallsetminus\left\{  0\right\}  }%
2C_{k}p^{\frac{-rN}{2}}\Omega(\Vert p^{r}x\Vert_{p})\cos{\left(  2\pi
\{p^{r-1}k\cdot x\}_{p}\right)  ,}%
\end{align*}
is also a complex-valued solution, where the bar denotes the complex
conjugate, see Remark \ref{Nota_3}. By taking the coefficient $C_{k}$ in
$\mathbb{R}$, we obtain a real-valued solution. Finally, the uniqueness of the
solution of the eigenvalue problem is established using the argument given in
the second part of the proof of the Lemma \ref{Lemma3}-(ii).
\end{proof}

\section{Energy levels and solutions of the $p$-adic Schr\"{o}dinger
equations}

\begin{remark}
\label{Nota_4}We use all the notation introduced in Remark \ref{Nota_3}. We
endow set $\mathbb{F}_{p}^{N}=\{0,1,\ldots,p-1\}^{N}$ with the lexicographic
order $\preceq_{\text{lex}}$, so $\left(  \mathbb{F}_{p}^{N},\preceq
_{\text{lex}}\right)  $ is a totally ordered set. Given $k=\left(
k_{1},\ldots,k_{N}\right)  \in\mathbb{F}_{p}^{N}\smallsetminus\left\{
0\right\}  $, there exists a unique $j=\left(  j_{1},\ldots,j_{N}\right)
\in\mathbb{F}_{p}^{N}\smallsetminus\left\{  0\right\}  $ such that $k+j=0$,
where the addition is the coordinatewise sum modulo $p$. We set%
\[
\mathbb{H}_{p}^{+}:=\left\{  k\in\mathbb{F}_{p}^{N}\smallsetminus\left\{
0\right\}  ;\text{there exits }j\in\mathbb{F}_{p}^{N}\smallsetminus\left\{
0\right\}  \text{ satisfying }k+j=0\text{ and }j\preceq_{\text{lex}}k\right\}
.
\]
Notice that $k,k^{\prime}\in\mathbb{H}_{p}^{+}$ implies $k+k^{\prime}\neq0$.
\end{remark}

We set
\[
I_{k,r}:=%
{\displaystyle\int\limits_{p^{-r}\mathbb{Z}_{p}^{N}}}
\cos^{2}{\left(  2\pi\{p^{r-1}k\cdot x\}_{p}\right)  }d^{N}x,
\]
and
\[
J_{k,j,r}:=%
{\displaystyle\int\limits_{p^{-r}\mathbb{Z}_{p}^{N}}}
\cos{\left(  2\pi\{p^{r-1}k\cdot x\}_{p}\right)  }\cos{\left(  2\pi
\{p^{r-1}j\cdot x\}_{p}\right)  }d^{N}x.
\]

\begin{lemma}
\label{Lemma5}With the above notation, $I_{k,r}=\frac{p^{rN}}{2}$, and for
$k\neq j$,
\[
J_{k,j,r}=\left\{
\begin{array}
[c]{cc}%
0 & k+j\neq0\\
\frac{p^{rN}}{2} & k+j=0.
\end{array}
\right.
\]

\end{lemma}

\begin{proof}
By using (\ref{Average}), with $b=0$,%
\begin{equation}%
{\displaystyle\int\limits_{p^{-r}\mathbb{Z}_{p}^{N}}}
\cos{\left(  2\pi\{p^{r-1}k\cdot x\}_{p}\right)  }d^{N}x=0, \label{Formula_1}%
\end{equation}%
\begin{equation}%
{\displaystyle\int\limits_{p^{-r}\mathbb{Z}_{p}^{N}}}
\sin{\left(  2\pi\{p^{r-1}k\cdot x\}_{p}\right)  }d^{N}x=0. \label{Formula_2}%
\end{equation}
Now, by the identity $\cos^{2}\left(  \alpha\right)  =\frac{1+\cos(2\alpha
)}{2}$, and $2\{p^{r-1}k\cdot x\}_{p}=\{p^{r-1}\left(  2k\right)  \cdot
x\}_{p}+m$, where $m$ is a integer,
\begin{align*}
I_{k,r}  &  =\frac{1}{2}%
{\displaystyle\int\limits_{p^{-r}\mathbb{Z}_{p}^{N}}}
\cos{\left(  4\pi\{p^{r-1}k\cdot x\}_{p}\right)  }d^{N}x+\frac{1}{2}%
{\displaystyle\int\limits_{p^{-r}\mathbb{Z}_{p}^{N}}}
d^{N}x\\
&  =\frac{1}{2}%
{\displaystyle\int\limits_{p^{-r}\mathbb{Z}_{p}^{N}}}
\cos{\left(  4\pi\{p^{r-1}k\cdot x\}_{p}\right)  }d^{N}x+\frac{p^{rN}}{2}\\
&  =\frac{1}{2}%
{\displaystyle\int\limits_{p^{-r}\mathbb{Z}_{p}^{N}}}
\cos{\left(  2\pi\{p^{r-1}\left(  2k\right)  \cdot x\}_{p}\right)  }%
d^{N}x+\frac{p^{rN}}{2}.
\end{align*}
Now, by the periodicity of the cosine function, we may assume that%
\[
l=2k=\left(  2k_{1}\text{ }\operatorname{mod}p,\ldots,2k_{N}\text{
}\operatorname{mod}p\right)  \in\mathbb{F}_{p}^{N}\smallsetminus\left\{
0\right\}  ,
\]
and by formula (\ref{Formula_1}), we have%
\[
I_{k,r}=\frac{p^{rN}}{2}.
\]

We now compute $J_{k,j,r}$ under the hypothesis $j\neq k$. By using the
identity%
\[
\cos\alpha\cos\beta=\frac{1}{2}\cos\left(  \alpha+\beta\right)  +\frac{1}%
{2}\cos\left(  \alpha-\beta\right)  ,
\]%
\begin{align*}
J_{k,j,r}  &  =\frac{1}{2}%
{\displaystyle\int\limits_{p^{-r}\mathbb{Z}_{p}^{N}}}
\cos{\left(  2\pi\left(  \{p^{r-1}k\cdot x\}_{p}+\{p^{r-1}j\cdot
x\}_{p}\right)  \right)  }d^{N}x+\\
&  \frac{1}{2}%
{\displaystyle\int\limits_{p^{-r}\mathbb{Z}_{p}^{N}}}
\cos{\left(  2\pi\left(  \{p^{r-1}k\cdot x\}_{p}-\{p^{r-1}j\cdot
x\}_{p}\right)  \right)  }\\
&  =\frac{1}{2}%
{\displaystyle\int\limits_{p^{-r}\mathbb{Z}_{p}^{N}}}
\cos{\left(  2\pi\{p^{r-1}\left(  k+j\right)  \cdot x\}_{p}\right)  }d^{N}x+\\
&  \frac{1}{2}%
{\displaystyle\int\limits_{p^{-r}\mathbb{Z}_{p}^{N}}}
\cos{\left(  2\pi\{p^{r-1}\left(  k-j\right)  \cdot x\}_{p}\right)  }d^{N}x\\
&  =:J_{k,j,r}^{\left(  1\right)  }+J_{k,j,r}^{\left(  2\right)  }.
\end{align*}
By the hypothesis $k-j\neq0$ and formula (\ref{Formula_1}), $J_{k,j,r}%
^{\left(  2\right)  }=0$. If $k+j\neq0$, $J_{k,j,r}^{\left(  1\right)  }=0$,
and if $k+j=0$, $J_{k,j,r}^{\left(  1\right)  }=\frac{p^{rN}}{2}$.
\end{proof}

\begin{theorem}
\label{Theorem1}The solutions to the \ eigenvalue problem%
\[
\left\{
\begin{array}
[c]{l}%
\Phi_{E}(x)\in L_{\mathbb{R}}^{2}\left(  p^{L}\mathbb{Z}_{p}^{N}\right)
;\text{ }\left\Vert \Phi_{E}\right\Vert _{2}=1\\
\\
m_{\alpha}\boldsymbol{D}^{\alpha}\Phi_{E}(x)=E\Phi_{E}(x)
\end{array}
\right.
\]
have the following form: for $E=E_{gnd}:=m_{\alpha}\frac{\left(
1-p^{-N}\right)  p^{-L\alpha}}{\left(  1-p^{-\alpha-N}\right)  }$,
$\Phi_{E_{gnd}}(x)=p^{\frac{LN}{2}}\Omega\left(  p^{L}\left\Vert x\right\Vert
_{p}\right)  $; and for $E_{r}=m_{\alpha}p^{(1-r)\alpha}$, with $r\leq-L$,
\begin{equation}
\Phi_{E_{r}}(x)=2p^{\frac{-rN}{2}}\Omega(\Vert p^{r}x\Vert_{p})%
{\displaystyle\sum\limits_{k\in\mathbb{H}_{p}^{+}}}
A_{k}\cos{\left(  2\pi\{p^{r-1}k\cdot x\}_{p}\right)  }, \label{Formula_Phi}%
\end{equation}
where the $A_{k}\in\mathbb{R}$ satisfy
\[
\sqrt{%
{\displaystyle\sum\limits_{k\in\mathbb{H}_{p}^{+}}}
A_{k}^{2}}=\frac{1}{\sqrt{2}}.
\]

\end{theorem}

\begin{proof}
Formula (\ref{Formula_Phi}) follows from \ Lemmas \ref{Lemma3}, \ref{Lemma4},
and Remark \ref{Nota_4}. We now compute the $L^{2}$-norm of $\Phi_{E_{r}}:$%
\begin{align*}
\left\Vert \Phi_{E_{r}}\right\Vert _{2}^{2}  &  =%
{\displaystyle\int\limits_{p^{L}\mathbb{Z}_{p}^{N}}}
\Phi_{E_{r}}^{2}(x)d^{N}x=4p^{-rN}%
{\displaystyle\int\limits_{p^{-r}\mathbb{Z}_{p}^{N}}}
\left(
{\displaystyle\sum\limits_{k\in\mathbb{H}_{p}^{+}}}
A_{k}\cos{\left(  2\pi\{p^{r-1}k\cdot x\}_{p}\right)  }\right)  ^{2}d^{N}x\\
&  =4p^{-rN}%
{\displaystyle\sum\limits_{k\in\mathbb{H}_{p}^{+}}}
A_{k}^{2}%
{\displaystyle\int\limits_{p^{-r}\mathbb{Z}_{p}^{N}}}
\cos^{2}{\left(  2\pi\{p^{r-1}k\cdot x\}_{p}\right)  }d^{N}x+\\
&  4p^{-rN}%
{\displaystyle\sum\limits_{\substack{k,j\in\mathbb{H}_{p}^{+}\\j\neq k}}}
\text{ }A_{k}A_{j}%
{\displaystyle\int\limits_{p^{-r}\mathbb{Z}_{p}^{N}}}
\cos{\left(  2\pi\{p^{r-1}k\cdot x\}_{p}\right)  }\cos{\left(  2\pi
\{p^{r-1}j\cdot x\}_{p}\right)  }d^{N}x.
\end{align*}
By Lemma \ref{Lemma5},%
\[
\left\Vert \Phi_{E_{r}}\right\Vert _{2}^{2}=2%
{\displaystyle\sum\limits_{k\in\mathbb{H}_{p}^{+}}}
A_{k}^{2}.
\]

\end{proof}

\begin{notation}
We denote by $\mathcal{C}\left(  K,\mathbb{R}\right)  $, the vector space of
real-valued functions defined on $K$.
\end{notation}

\begin{theorem}
\label{Theorem2}The initial value problem%
\begin{equation}
\left\{
\begin{array}
[c]{l}%
\Psi(x,t)\in L_{\mathbb{R}}^{2}\left(  p^{L}\mathbb{Z}_{p}^{N}\right)  \text{
for }t\geq0\text{ fixed}\\
\\
\Psi(x,t)\in C^{1}\left(  0,\infty\right)  \text{ for }x\in p^{L}%
\mathbb{Z}_{p}^{N}\text{ fixed}\\
\\
i\hbar\frac{\partial\Psi(x,t)}{\partial t}=\left\{  m_{\alpha}\boldsymbol{D}%
^{\alpha}+V(\left\Vert x\right\Vert _{p})\right\}  \Psi(x,t)\\
\\
\Psi(x,0)=\Psi(-x,0)\in\mathcal{C}\left(  p^{L}\mathbb{Z}_{p}^{N}%
,\mathbb{R}\right)
\end{array}
\right.  \label{Cauchy_Problem_3}%
\end{equation}
has a unique solution of the form%
\begin{gather}
\Psi(x,t)=A_{0}p^{\frac{LN}{2}}\Omega\left(  p^{L}\left\Vert x\right\Vert
_{p}\right)  e^{\frac{-i}{\hbar}E_{gnd}t}\label{Wave_Function_1}\\
+%
{\displaystyle\sum\limits_{k\in\mathbb{H}_{p}^{+}}}
\text{ }%
{\displaystyle\sum\limits_{r\leq-L}}
\text{ }2p^{\frac{-rN}{2}}A_{kr}e^{\frac{-i}{\hbar}E_{r}t}\Omega(\Vert
p^{r}x\Vert_{p})\cos{\left(  2\pi\{p^{r-1}k\cdot x\}_{p}\right)  ,}\nonumber
\end{gather}
where the $A_{0}$, $A_{kr}\in\mathbb{R}$.
\end{theorem}

\begin{proof}
The result follows from the expansion (\ref{General_form}), by Lemmas
\ref{Lemma0}, \ref{Lemma2}, and Theorem \ref{Theorem1}.
\end{proof}

\section{CTQWs and $p$-adic QM}

\subsection{Construction of CTQWs I}

This section presents two different constructions of CTQWs based on $p$-adic
QM. The first construction is motivated by \cite{Mulkne-Blumen}. We take
$\hbar=1$.

Let $\boldsymbol{H}$ be a self-adjoint Hamiltonian, with a pure point
spectrum:%
\[
\boldsymbol{H}\psi_{n}=E_{n}\psi_{n}\text{ for }n\in\mathbb{N}=\left\{
0,1,\ldots,n,n+1,\ldots\right\}  ,
\]
such that $\left\{  \psi_{n}\right\}  _{n\in\mathbb{N}}$ is an orthonormal
basis of $L^{2}(\mathcal{K})$, where $\mathcal{K}$ is an open compact subset
of $\mathbb{Q}_{p}^{N}$, for instance $\mathcal{K}=p^{L}\mathbb{Z}_{p}^{N}$.
Let $\mathcal{N}\subset\mathcal{K}$ be a subset of measure zero, and let
$\mathcal{K}_{j}$ be disjoint, open, compact subsets such that
\begin{equation}
\mathcal{K}\smallsetminus\mathcal{N}=%
{\displaystyle\bigsqcup\limits_{j\in J}}
\mathcal{K}_{j}, \label{partition_K}%
\end{equation}
where $J$ is countable (i.e., $J$ is either finite or countably infinite).

The function
\begin{equation}
\Psi_{v}\left(  x,t\right)  :=e^{-i\boldsymbol{H}t}\left(  c_{v}1_{K_{v}%
}(x)\right)  \in L^{2}\left(  \mathcal{K}\right)  , \label{Function_Psi}%
\end{equation}
where $\left\Vert c_{v}1_{\mathcal{K}_{v}}\right\Vert _{2}=1$, is the solution
of the Cauchy problem:%
\[
\left\{
\begin{array}
[c]{ll}%
\Psi\left(  \cdot,t\right)  \in L^{2}(\mathcal{K}), & t\geq0\\
& \\
i\frac{\partial\Psi\left(  x,t\right)  }{\partial t}=\boldsymbol{H}\Psi\left(
x,t\right)  , & x\in\mathcal{K},t>0\\
& \\
\Psi\left(  x,0\right)  =c_{v}1_{\mathcal{K}_{v}}\left(  x\right)  . &
\end{array}
\right.
\]
We construct a network with nodes $j\in J$, where $j$ is identified \ with
$1_{\mathcal{K}_{j}}$, and define the transition probability between nodes $v$
and $r$ as%
\[
\widetilde{\pi}_{r,v}\left(  t\right)  =\left\vert \left\langle c_{r}%
1_{\mathcal{K}_{r}},e^{-i\boldsymbol{H}t}\left(  c_{v}1_{\mathcal{K}_{v}%
}\right)  \right\rangle \right\vert ^{2}=\left\vert \left\langle
c_{r}1_{\mathcal{K}_{r}},\Psi_{v}\right\rangle \right\vert ^{2},
\]
where $\left\langle \cdot,\cdot\right\rangle $\ is the inner product of
$L^{2}\left(  \mathcal{K}\right)  $. These transition probabilities are
amenable for computing; however, we do not know if
\[%
{\displaystyle\sum\limits_{r\in J}}
\widetilde{\pi}_{r,v}\left(  t\right)  =1.
\]
for the Hamiltonian $\boldsymbol{H}$ of the infinite wall potentials.

\subsection{Construction of CTQWs II}

The second construction is based on the Born interpretation of the
wavefunctions, which is also valid in the $p$-adic framework, \cite{Zuniga-AP}%
-\cite{Zuniga-PhA}. Let $\mathcal{K}$ be an open compact subset of
$\mathbb{Q}_{p}^{N}$ as before. The function $\left\vert \Psi_{v}\left(
x,t\right)  \right\vert ^{2}$ is the probability density for a particle/walker
at the time $t$ given that at $t=0$ the particle/walker was in $\mathcal{K}%
_{v}$. By the Born interpretation%
\[
\pi_{r,v}\left(  t\right)  :=%
{\displaystyle\int\limits_{\mathcal{K}_{r}}}
\left\vert \Psi_{v}\left(  x,t\right)  \right\vert ^{2}dx
\]
is the probability of finding the particle/walker in $\mathcal{K}_{r}$; which
we interpret as a transition probability from a $\mathcal{K}_{v}$ to
$\mathcal{K}_{r}$. Now, using partition (\ref{partition_K}), we have
\begin{align*}
1  &  =%
{\displaystyle\int\limits_{\mathcal{K}}}
\left\vert \Psi_{v}\left(  x,t\right)  \right\vert ^{2}dx=%
{\displaystyle\int\limits_{\mathcal{K}\smallsetminus\mathcal{N}}}
\left\vert \Psi_{v}\left(  x,t\right)  \right\vert ^{2}dx\\
&  =%
{\displaystyle\sum\limits_{j\in J}}
\text{ }%
{\displaystyle\int\limits_{\mathcal{K}_{j}}}
\left\vert \Psi_{v}\left(  x,t\right)  \right\vert ^{2}dx=%
{\displaystyle\sum\limits_{j\in J}}
\pi_{r,v}\left(  t\right)  .
\end{align*}
So, the matrix $\left[  \pi_{r,v}\left(  t\right)  \right]  _{r,v\in
\mathbb{N}}$ defines a quantum Markov chain with space state $\mathbb{N}$,
where a walker jumps between the sets $\mathcal{K}_{v}$ to $\mathcal{K}_{r}$
with a transition probability $\pi_{r,v}\left(  t\right)  $.

The computation of the transition probabilities $\pi_{r,v}\left(  t\right)  $
is involved, even in dimension one and for simple partitions of the unit ball.
In the following sections, we explicitly compute some examples in dimension one.

\subsection{The Farhi-Gutmann CTQW}

The CTQWs on graphs play a central role in quantum computing. In this section,
we review the Farhi-Gutmann CTQWs, \cite{Farhi-Gutman}, and compare them with
the ones introduced here.

Let $\mathcal{G}$ be an undirected graph with vertices $I=1,2,\ldots,N$, and
adjacency matrix $\left[  A_{JI}\right]  $,
\[
A_{JI}:=\left\{
\begin{array}
[c]{ll}%
1 & \text{if the vertices }J\text{ and }I\text{ are connected}\\
& \\
0 & \text{otherwise.}%
\end{array}
\right.
\]
We assume that this matrix is symmetric. We attach to every vertex $I$ and a
vector $e_{I}\in\mathbb{C}^{N}$, so $\left\{  e_{I};I=1,2,\ldots,N\right\}  $
is an orthonormal basis of $\mathbb{C}^{N}$ as a Hilbert space with inner
proiduct $\left\langle \cdot,\cdot\right\rangle $.

We set%
\[
\left\langle e_{I}\right\vert H\left\vert e_{J}\right\rangle =\left\{
\begin{array}
[c]{lll}%
-\gamma & \text{if} & J\neq I\text{ and }A_{JI}=1\\
&  & \\
0 & \text{if} & J\neq I\text{ and }A_{JI}=0\\
&  & \\
\mathrm{val}(I)\gamma & \text{if} & J=I,
\end{array}
\right.
\]
where $\gamma$ denotes the jumping rate per unit of time between vertices, and
$\mathrm{val}(I)$ denotes the valence of $I$, i.e. the number of connections
from $I$ to its other vertices. The Schr\"{o}dinger equation of the
Farhi-Gutmann CTQW is%
\[
i\frac{\partial}{\partial t}\left\langle e_{I}\right\vert \left.  \Psi\left(
t\right)  \right\rangle =%
{\displaystyle\sum\limits_{K=1}^{N}}
\left\langle e_{I}\right\vert H\left\vert e_{K}\right\rangle \left\langle
e_{K}\right\vert \left.  \Psi\left(  t\right)  \right\rangle ,
\]
for $I=1,2,\ldots,N$.

The Farhi-Gutmann CTQWs depend entirely on an the adjacency matrices of the
graphs. Then, CTQWs introduced here cannot be obtained from this construction.
So, they are entirely new. The drawback of our construction is that we need
the wavefunction, and also we have to compute some integrals (the $c_{r}$
constants), these calculations are in general difficult.

\section{Continuous-time quantum walks I}

\begin{proposition}
\label{Prop5}Let $f\in L_{\mathbb{R}}^{2}(p^{L}\mathbb{Z}_{p}^{N})$ be a
continuous function satisfying $f\left(  x\right)  =f\left(  \left\Vert
x\right\Vert _{p}\right)  $. Then
\[
f\left(  x\right)  =C_{0}p^{\frac{LN}{2}}\Omega\left(  p^{L}\left\Vert
x\right\Vert _{p}\right)  +%
{\displaystyle\sum\limits_{k}}
\text{ }%
{\displaystyle\sum\limits_{r\leq-L}}
\text{ \ }C_{r0k}p^{\frac{-rN}{2}}\Omega(\Vert p^{r}x\Vert_{p})\cos{\left(
2\pi\{p^{r-1}k\cdot x\}_{p}\right)  ,}%
\]
where the coefficients $C_{0}$, $C_{r0k}$ are real numbers. Furthermore,%
\begin{equation}
\left\Vert f\right\Vert _{2}^{2}=C_{0}^{2}+%
{\displaystyle\sum\limits_{k}}
\text{ }%
{\displaystyle\sum\limits_{r\leq-L}}
\text{ \ }C_{r0k}^{2}. \label{norm}%
\end{equation}

\end{proposition}

\begin{proof}
Take $f\in L_{\mathbb{R}}^{2}(p^{L}\mathbb{Z}_{p}^{N})\subset L^{2}%
(p^{L}\mathbb{Z}_{p}^{N})$, then using the orthonormal basis
(\ref{General_form}),%
\begin{align}
f\left(  x\right)   &  =C_{0}p^{\frac{LN}{2}}\Omega\left(  p^{L}\left\Vert
x\right\Vert _{p}\right)  +\label{Fourier_series}\\
&
{\displaystyle\sum\limits_{k}}
\text{ }%
{\displaystyle\sum\limits_{r\leq-L}}
\text{ \ }%
{\displaystyle\sum\limits_{b\in p^{L+r}\mathbb{Z}_{p}^{N}}}
C_{rbk}\Psi_{rbk}\left(  x\right)  ,\nonumber
\end{align}
with $C_{0}\in\mathbb{R}$, $C_{rbk}\in\mathbb{C}$. We now compute $C_{rbk}$.
If $b\neq0$,
\[
C_{rbk}=%
{\displaystyle\int\limits_{p^{-r}b+p^{-r}\mathbb{Z}_{p}^{N}}}
f\left(  \left\Vert x\right\Vert _{p}\right)  \overline{\Psi_{rbk}\left(
x\right)  }d^{N}x=f\left(  \left\Vert p^{-r}b\right\Vert _{p}\right)
\overline{%
{\displaystyle\int\limits_{p^{-r}b+p^{-r}\mathbb{Z}_{p}^{N}}}
\Psi_{rbk}\left(  x\right)  d^{N}x}=0.
\]

If $b=0$,%
\[
C_{r0k}=%
{\displaystyle\int\limits_{p^{-r}\mathbb{Z}_{p}^{N}}}
f\left(  \left\Vert x\right\Vert _{p}\right)  \overline{\Psi_{r0k}\left(
x\right)  }d^{N}x\in\mathbb{R}%
\]
because $f\left(  \left\Vert x\right\Vert _{p}\right)  =f\left(  \left\Vert
-x\right\Vert _{p}\right)  $.

Then, the coefficients $C_{0}$, $C_{r0k}$ are real, and taking real parts on
both sides of (\ref{Fourier_series}),
\begin{align*}
f\left(  x\right)   &  =C_{0}p^{\frac{LN}{2}}\Omega\left(  p^{L}\left\Vert
x\right\Vert _{p}\right)  +\\
&
{\displaystyle\sum\limits_{k}}
\text{ }%
{\displaystyle\sum\limits_{r\leq-L}}
\text{ \ }C_{r0k}p^{\frac{-rN}{2}}\Omega(\Vert p^{r}x\Vert_{p})\cos{\left(
2\pi\{p^{r-1}k\cdot x\}_{p}\right)  .}%
\end{align*}
Furthermore, from (\ref{Fourier_series}) it follows that the $L^{2}$-norm of
$f$ is as in (\ref{norm}).
\end{proof}

\subsection{The Fourier series of $1_{S_{-j}}(x)$}

In this section, we use Proposition \ref{Prop5} in dimension $N=1$. For the
sake of simplicity we take, $L=0$, $m=-r$, and $m\geq0$, then, for $h(x)\in
L_{\mathbb{R}}^{2}(\mathbb{Z}_{p})\cap\mathcal{C}(\mathbb{Z}_{p},\mathbb{R})$,
with $h(x)=h\left(  \left\vert x\right\vert _{p}\right)  $,%
\begin{equation}
h\left(  x\right)  =C_{0}\Omega\left(  \left\vert x\right\vert _{p}\right)  +%
{\displaystyle\sum\limits_{k}}
{\displaystyle\sum\limits_{m=0}^{\infty}}
\ C_{\left(  -m\right)  0k}\Psi_{\left(  -m\right)  0k}\left(  x\right)  ,
\label{Fourier_series_2}%
\end{equation}
i.e.,%
\[
h\left(  x\right)  =C_{0}\Omega\left(  \left\vert x\right\vert _{p}\right)  +%
{\displaystyle\sum\limits_{k}}
{\displaystyle\sum\limits_{m=0}^{\infty}}
\ C_{\left(  -m\right)  0k}p^{\frac{m}{2}}\Omega(\left\vert p^{-m}x\right\vert
_{p})\cos{\left(  2\pi\{p^{-1-m}kx\}_{p}\right)  }%
\]
We also set
\[
E_{gnd}:=m_{\alpha}\frac{\left(  1-p^{-1}\right)  }{\left(  1-p^{-\alpha
-1}\right)  },\text{ \ }E_{m}:=m_{\alpha}p^{(1+m)\alpha}\text{ for }m\geq0.
\]
We denote by $1_{S_{-j}}(x)$, the characteristic function of the sphere%
\[
S_{-j}=\left\{  x\in\mathbb{Z}_{p};\left\vert x\right\vert _{p}=p^{-j}%
\right\}  =p^{j}S_{0}=p^{j}\mathbb{Z}_{p}^{\times},\text{ for }j\geq0.
\]

\begin{lemma}
\label{Lemma6}The Fourier series of $1_{S_{-j}}(x)$ is given by%
\begin{align*}
1_{S_{-j}}(x)  &  =C_{0}\Omega\left(  \left\vert x\right\vert _{p}\right)  +%
{\displaystyle\sum\limits_{k\in\mathbb{F}_{p}\smallsetminus\left\{  0\right\}
}}
\text{ }%
{\displaystyle\sum\limits_{m=0}^{j}}
C_{\left(  -m\right)  0k}\Psi_{\left(  -m\right)  0k}\left(  x\right) \\
&  =C_{0}\Omega\left(  \left\vert x\right\vert _{p}\right)  +%
{\displaystyle\sum\limits_{k\in\mathbb{F}_{p}\smallsetminus\left\{  0\right\}
}}
\text{ }%
{\displaystyle\sum\limits_{m=0}^{j}}
\text{ }p^{\frac{m}{2}}C_{mk}\Omega(p^{m}\left\vert x\right\vert _{p}%
)\cos{\left(  2\pi\{p^{-m-1}kx\}_{p}\right)  ,}%
\end{align*}
where%
\begin{align*}
C_{0}  &  =p^{-j}\left(  1-p^{-1}\right)  \text{, and }\\
C_{\left(  -m\right)  0k}  &  :=C_{mk}=\left\{
\begin{array}
[c]{lll}%
0 & \text{if} & m>j\\
&  & \\
-p^{-\frac{j}{2}-1} & \text{if} & m=j\\
&  & \\
p^{\frac{m}{2}-j}\left(  1-p^{-1}\right)  & \text{if} & m\leq j-1.
\end{array}
\right.
\end{align*}

\end{lemma}

\begin{proof}
We compute the coefficients in the Fourier series (\ref{Fourier_series_2}) for
$1_{S_{-j}}(x)$. We first compute $C_{0}$:%
\[
C_{0}=%
{\displaystyle\int\limits_{\mathbb{Z}_{p}}}
1_{S_{-j}}(x)\Omega\left(  \left\vert x\right\vert _{p}\right)  dx=%
{\displaystyle\int\limits_{p^{j}\mathbb{Z}_{p}^{\times}}}
dx=p^{-j}\left(  1-p^{-1}\right)  .
\]

We now compute $C_{\left(  -m\right)  0k}:=C_{mk}$. By using that
$p^{j}\mathbb{Z}_{p}^{\times}\subset p^{m}\mathbb{Z}_{p}$ if and only if
$j\geq m$, we have
\[
p^{j}\mathbb{Z}_{p}^{\times}\cap p^{m}\mathbb{Z}_{p}=\left\{
\begin{array}
[c]{lll}%
p^{j}\mathbb{Z}_{p}^{\times} & \text{if} & j\geq m\\
&  & \\
\varnothing & \text{if} & j<m.
\end{array}
\right.
\]
Now, using that $C_{mk}\in\mathbb{R}$,
\begin{align*}
C_{mk}  &  =%
{\displaystyle\int\limits_{\mathbb{Z}_{p}}}
1_{S_{-j}}(x)\overline{\Psi_{\left(  -m\right)  0k}\left(  x\right)  }dx=%
{\displaystyle\int\limits_{p^{j}\mathbb{Z}_{p}^{\times}\cap p^{m}%
\mathbb{Z}_{p}}}
\overline{\Psi_{\left(  -m\right)  0k}\left(  x\right)  }dx\\
&  =p^{\frac{m}{2}}%
{\displaystyle\int\limits_{p^{j}\mathbb{Z}_{p}^{\times}\cap p^{m}%
\mathbb{Z}_{p}}}
\cos{\left(  2\pi\{p^{-m-1}kx\}_{p}\right)  }dx.
\end{align*}
Then $C_{mk}=0$ if $j<m$. If $j\geq m$, changing variables as $x=p^{j}y$,
$dx=p^{-j}dy$,
\[
C_{mk}=p^{\frac{m}{2}}%
{\displaystyle\int\limits_{p^{j}\mathbb{Z}_{p}^{\times}}}
\cos{\left(  2\pi\{p^{-m-1}kx\}_{p}\right)  }dx=p^{\frac{m}{2}-j}%
{\displaystyle\int\limits_{\mathbb{Z}_{p}^{\times}}}
\cos{\left(  2\pi\{p^{-m-1+j}ky\}_{p}\right)  }dy.
\]
Now, two cases occur: (i) $j\geq m+1$, (ii) $j=m$. In the first case,
$j-m-1\geq0$, and $p^{-m-1+j}ky\in\mathbb{Z}_{p}$ for any $y\in\mathbb{Z}%
_{p}^{\times}$ and $k\in\mathbb{F}_{p}\smallsetminus\left\{  0\right\}  $,
consequently, $\{p^{-m-1+j}ky\}_{p}=0$, and
\[
C_{mk}=p^{\frac{m}{2}-j}\left(  1-p^{-1}\right)  \text{, for }m\leq
j-1\text{.}%
\]
In the case $j=m$, by using the partition $\mathbb{Z}_{p}^{\times}=%
{\textstyle\bigsqcup\nolimits_{l=1}^{p-1}}
\left(  l+p\mathbb{Z}_{p}\right)  $,
\[
C_{jk}=p^{-\frac{j}{2}}%
{\displaystyle\int\limits_{\mathbb{Z}_{p}^{\times}}}
\cos{\left(  2\pi\{p^{-1}ky\}_{p}\right)  }dy=p^{-\frac{j}{2}}%
{\displaystyle\sum\limits_{l=1}^{p-1}}
\text{ }%
{\displaystyle\int\limits_{l+p\mathbb{Z}_{p}}}
\cos{\left(  2\pi\{p^{-1}ky\}_{p}\right)  }dy.
\]
By using $p^{-1}ky=p^{-1}kl+\mathbb{Z}_{p}$, then $\left\{  p^{-1}ky\right\}
_{p}=p^{-1}kl$, then
\[
\cos{\left(  2\pi\{p^{-1}ky\}_{p}\right)  \mid}_{l+p\mathbb{Z}_{p}}{=}%
\cos{\left(  2\pi\{p^{-1}kl\}_{p}\right)  =}\cos{\left(  2\pi p^{-1}kl\right)
,}%
\]
and consequently,%
\[
C_{jk}=p^{-\frac{j}{2}-1}%
{\displaystyle\sum\limits_{l=1}^{p-1}}
\cos{\left(  2\pi p^{-1}kl\right)  =-}p^{-\frac{j}{2}-1}.
\]
The last equality follows from the well-known formula:%
\[%
{\displaystyle\sum\limits_{k=1}^{p-1}}
e^{2\pi i\frac{kl}{p}}=-1,\text{ for }l\text{ not divisible by }p\text{.}%
\]

\end{proof}

\subsection{The graph $K_{\infty}$}

We now define an infinite simple graph $\left(  K_{\infty},E,V\right)  $,
where the set of vertices $V=\mathbb{N}$, and the adjacency matrix $A=\left[
A_{m,j}\right]  _{m,j\in\mathbb{N}}$ is defined as
\[
A_{m,j}=\left\{
\begin{array}
[c]{ccc}%
0 & \text{if} & m=j\\
&  & \\
1 & \text{if} & m\neq j.
\end{array}
\right.
\]
Thus, if $A_{m,j}=1$, there exists an undirected edge in $E$ connecting $m$
and $j$. We identify the vertex $j$ with $S_{-j}$. The graph $K_{\infty}$ is
fully connected graph, which is a simple undirected graph where every pair of
distinct vertices is connected by a unique edge

\subsection{Transition probabilities}

The transition probabilities $\pi_{r,v}\left(  t\right)  $ are computed using
the formula%
\begin{equation}
\pi_{r,v}\left(  t\right)  =%
{\displaystyle\int\limits_{S_{-r}}}
\left\vert 1_{S_{-r}}(x)\Psi_{v}\left(  x,t\right)  \right\vert ^{2}dx=p^{-r}%
{\displaystyle\int\limits_{S_{0}}}
\left\vert 1_{S_{-r}}(p^{r}u)\Psi_{v}\left(  p^{r}u,t\right)  \right\vert
^{2}du, \label{Formula_10}%
\end{equation}
where%
\begin{align*}
\Psi_{v}\left(  x,t\right)   &  =p^{\frac{v}{2}}\left(  1-p^{-1}\right)
^{\frac{-1}{2}}e^{-i\boldsymbol{H}t}1_{S_{-v}}(x)\\
&  =p^{\frac{v}{2}}\left(  1-p^{-1}\right)  ^{\frac{-1}{2}}\left\{
C_{0}e^{-iE_{gnd}t}\Omega\left(  \left\vert x\right\vert _{p}\right)  +%
{\displaystyle\sum\limits_{k\in\mathbb{F}_{p}\smallsetminus\left\{  0\right\}
}}
\text{ }%
{\displaystyle\sum\limits_{m=0}^{v}}
C_{\left(  -m\right)  0k}e^{-iE_{m}t}\Psi_{\left(  -m\right)  0k}\left(
x\right)  \right\}  ,
\end{align*}
see Lemma \ref{Lemma6}; then%
\begin{align*}
1_{S_{-r}}(x)\Psi_{v}\left(  x,t\right)   &  =p^{\frac{v}{2}}\left(
1-p^{-1}\right)  ^{\frac{-1}{2}}\times\\
&  \left\{  C_{0}e^{-iE_{gnd}t}1_{S_{-r}}(x)+%
{\displaystyle\sum\limits_{k\in\mathbb{F}_{p}\smallsetminus\left\{  0\right\}
}}
\text{ }%
{\displaystyle\sum\limits_{m=0}^{v}}
C_{\left(  -m\right)  0k}e^{-iE_{m}t}1_{S_{-r}}(x)\Psi_{\left(  -m\right)
0k}\left(  x\right)  \right\}  .
\end{align*}
By using that \textrm{supp}$\Psi_{\left(  -m\right)  0k}\left(  x\right)
=p^{m}\mathbb{Z}_{p}$, with $\Psi_{\left(  -m\right)  0k}\left(  x\right)
=p^{\frac{m}{2}}\chi_{p}(p^{-m-1}kx)\Omega\left(  p^{m}\left\vert x\right\vert
_{p}\right)  $, we get%
\begin{equation}
1_{S_{-r}}(x)\Psi_{\left(  -m\right)  0k}\left(  x\right)  =\left\{
\begin{array}
[c]{lll}%
p^{\frac{m}{2}}1_{S_{-r}}(x)\chi_{p}(p^{-m-1}kx) & \text{if} & r\geq m\\
&  & \\
0 & \text{if} & r<m,
\end{array}
\right.  \label{Restriction_2}%
\end{equation}
and%
\begin{align*}
1_{S_{-r}}(x)\Psi_{v}\left(  x,t\right)   &  =p^{\frac{v}{2}}\left(
1-p^{-1}\right)  ^{\frac{-1}{2}}\times\\
&  \left\{  C_{0}e^{-iE_{gnd}t}1_{S_{-r}}(x)+%
{\displaystyle\sum\limits_{k\in\mathbb{F}_{p}\smallsetminus\left\{  0\right\}
}}
{\displaystyle\sum\limits_{m=0}^{\min\left\{  v,r\right\}  }}
p^{\frac{m}{2}}C_{\left(  -m\right)  0k}e^{-iE_{m}t}1_{S_{-r}}(x)\chi
_{p}(p^{-m-1}kx)\right\}  .
\end{align*}

By replacing $1_{S_{-r}}(x)\Psi_{v}\left(  x,t\right)  $ in (\ref{Formula_10}%
), we conclude that%
\begin{align*}
\pi_{r,v}\left(  t\right)   &  =p^{v-r}\left(  1-p^{-1}\right)  ^{-1}\times\\
&
{\displaystyle\int\limits_{S_{0}}}
\left\vert C_{0}e^{-iE_{gnd}t}+%
{\displaystyle\sum\limits_{k\in\mathbb{F}_{p}\smallsetminus\left\{  0\right\}
\ }}
{\displaystyle\sum\limits_{m=0}^{\min\left\{  v,r\right\}  \ }}
p^{\frac{m}{2}}C_{\left(  -m\right)  0k}e^{-iE_{m}t}\chi_{p}(p^{-m-1+r}%
ku)\right\vert ^{2}du.
\end{align*}

\begin{theorem}
The transition probability between the nodes $r$ and $v$ is given by
\begin{align*}
\pi_{r,v}\left(  t\right)   &  =p^{v-r}\left(  1-p^{-1}\right)  ^{-1}\times\\
&
{\displaystyle\int\limits_{S_{0}}}
\left\vert C_{0}e^{-iE_{gnd}t}+(p-1)%
{\displaystyle\sum\limits_{m=0}^{\min\left\{  v,r\right\}  \ }}
p^{\frac{m}{2}}C_{\left(  -m\right)  0k}e^{-iE_{m}t}\chi_{p}(p^{-m-1+r}%
u)\right\vert ^{2}du,
\end{align*}
where the constants $C_{0}$, $C_{\left(  -m\right)  0k}$ are given in Lemma
\ref{Lemma6}.
\end{theorem}

\section{Continuous-time quantum walks II}

\subsection{The Fourier series of $\Omega\left(  p^{R_{0}}\left\vert
x\right\vert _{p}\right)  $}

\begin{lemma}
\label{Lemma7}Let $R_{0}$ be a positive integer. The Fourier series of
$\Omega\left(  p^{R_{0}}\left\vert x\right\vert _{p}\right)  $ is given by%
\begin{align*}
\Omega\left(  p^{R_{0}}\left\vert x\right\vert _{p}\right)   &  =C_{0}%
^{\prime}\Omega\left(  \left\vert x\right\vert _{p}\right)  +%
{\displaystyle\sum\limits_{k\in\mathbb{F}_{p}\smallsetminus\left\{  0\right\}
}}
\text{ }%
{\displaystyle\sum\limits_{m=0}^{R_{0}-1}}
C_{\left(  -m\right)  0k}^{\prime}\Psi_{\left(  -m\right)  0k}\left(  x\right)
\\
&  =C_{0}^{\prime}\Omega\left(  \left\vert x\right\vert _{p}\right)  +%
{\displaystyle\sum\limits_{k\in\mathbb{F}_{p}\smallsetminus\left\{  0\right\}
}}
\text{ }%
{\displaystyle\sum\limits_{m=0}^{R_{0}-1}}
\text{ }p^{\frac{m}{2}}C_{mk}^{\prime}\Omega(p^{m}\left\vert x\right\vert
_{p})\cos{\left(  2\pi\{p^{-m-1}kx\}_{p}\right)  ,}%
\end{align*}
where%
\begin{align*}
C_{0}^{\prime}  &  =p^{-R_{0}}\text{, and }\\
C_{\left(  -m\right)  0k}^{\prime}  &  :=C_{mk}^{\prime}=\left\{
\begin{array}
[c]{lll}%
0 & \text{if} & m\geq R_{0}\\
&  & \\
p^{\frac{m}{2}-R_{0}} & \text{if} & m\leq R_{0}-1.
\end{array}
\right.
\end{align*}

\end{lemma}

\begin{proof}
We compute the coefficients in the Fourier series (\ref{Fourier_series_2}) for
$\Omega\left(  p^{R_{0}}\left\vert x\right\vert _{p}\right)  $. We use the
notation $C_{0}^{\prime}$, $C_{\left(  -m\right)  0k}^{\prime}$. We now
compute $C_{\left(  -m\right)  0k}^{\prime}:=C_{mk}^{\prime}$:%
\[
C_{mk}^{\prime}=%
{\displaystyle\int\limits_{\mathbb{Q}_{p}}}
\Omega\left(  p^{R_{0}}\left\vert x\right\vert _{p}\right)  \overline
{\Psi_{\left(  -m\right)  0k}\left(  x\right)  }dx{=}%
{\displaystyle\int\limits_{p^{R_{0}}\mathbb{Z}_{p}\cap p^{m}\mathbb{Z}_{p}}}
\overline{\Psi_{\left(  -m\right)  0k}\left(  x\right)  }dx.
\]

If $m\geq R_{0}$, by using (\ref{Average}) $C_{mk}^{\prime}=0$. If $m\leq
R_{0}-1$,
\[
C_{mk}^{\prime}=%
{\displaystyle\int\limits_{p^{R_{0}}\mathbb{Z}_{p}}}
\overline{\Psi_{\left(  -m\right)  0k}\left(  x\right)  }dx{=}%
{\displaystyle\int\limits_{\mathbb{Q}_{p}}}
p^{\frac{m}{2}}\Omega\left(  p^{R_{0}}\left\vert x\right\vert _{p}\right)
dx=p^{\frac{m}{2}-R_{0}}.
\]

\end{proof}

\subsection{The graph $K_{R_{0}}$}

From now on, we assume that $R_{0}\geq3$. \ We denote by $\left(  K_{R_{0}%
},E,V\right)  $, the fully connected graph, also known as a complete graph,
which is a simple undirected graph where a unique edge connects any two
vertices. We identify the set of vertices $V$ with $\left\{  0,1,\ldots
,R_{0}-1\right\}  \cup\left\{  \Gamma\right\}  $. A vertex $j\in\left\{
0,1,\ldots,R_{0}-1\right\}  $ corresponds to a sphere $S_{-j}$, while $\Gamma$
corresponds to $p^{R_{0}}\mathbb{Z}_{p}$. We introduce a CTQW where the walker
jumps randomly between the sets $S_{0},S_{-1},\ldots,S_{-\left(
R_{0}-1\right)  }$, $p^{R_{0}}\mathbb{Z}_{p}$.

\subsection{Transition probabilities}

\subsubsection{Computation of $\pi_{\Gamma,\Gamma}\left(  t\right)  $}

The transition probability $\pi_{\Gamma,\Gamma}\left(  t\right)  $ is given by%
\[
\pi_{\Gamma,\Gamma}\left(  t\right)  =%
{\displaystyle\int\limits_{p^{R_{0}}\mathbb{Z}_{p}}}
\left\vert \Omega\left(  p^{R_{0}}\left\vert x\right\vert _{p}\right)
\Psi_{\Gamma}\left(  x,t\right)  \right\vert ^{2}dx,
\]
where%
\begin{align*}
\Psi_{\Gamma}\left(  x,t\right)   &  =p^{\frac{R_{0}}{2}}e^{-i\boldsymbol{H}%
t}\Omega\left(  p^{R_{0}}\left\vert x\right\vert _{p}\right)  =\\
&  p^{\frac{R_{0}}{2}}\left\{  C_{0}^{\prime}e^{-iE_{gnd}t}\Omega\left(
\left\vert x\right\vert _{p}\right)  +%
{\displaystyle\sum\limits_{k\in\mathbb{F}_{p}\smallsetminus\left\{  0\right\}
}}
\text{ }%
{\displaystyle\sum\limits_{m=0}^{R_{0}-1}}
C_{\left(  -m\right)  0k}^{\prime}e^{-iE_{m}t}\Psi_{\left(  -m\right)
0k}\left(  x\right)  \right\}  ,
\end{align*}
and
\begin{gather*}
\Omega\left(  p^{R_{0}}\left\vert x\right\vert _{p}\right)  \Psi_{\Gamma
}\left(  x,t\right)  =\\
p^{\frac{R_{0}}{2}}\left\{  C_{0}^{\prime}e^{-iE_{gnd}t}\Omega\left(
p^{R_{0}}\left\vert x\right\vert _{p}\right)  +%
{\displaystyle\sum\limits_{k\in\mathbb{F}_{p}\smallsetminus\left\{  0\right\}
}}
\text{ }%
{\displaystyle\sum\limits_{m=0}^{R_{0}-1}}
C_{\left(  -m\right)  0k}^{\prime}e^{-iE_{m}t}\Omega\left(  p^{R_{0}%
}\left\vert x\right\vert _{p}\right)  \Psi_{\left(  -m\right)  0k}\left(
x\right)  \right\}  ,
\end{gather*}
see Lemma \ref{Lemma7}. Now, by (\ref{Restriction}),
\[
\Omega\left(  p^{R_{0}}\left\vert x\right\vert _{p}\right)  \Psi_{\left(
-m\right)  0k}\left(  x\right)  =\left\{
\begin{array}
[c]{lll}%
\Psi_{\left(  -m\right)  0k}\left(  x\right)   & \text{if} & \text{ }m\geq
R_{0}\\
&  & \\
p^{\frac{m}{2}}\Omega\left(  p^{R_{0}}\left\vert x\right\vert _{p}\right)   &
\text{if} & \text{ \ }m\leq R_{0}-1,
\end{array}
\right.
\]
and Lemma \ref{Lemma7},%
\begin{gather*}
\Omega\left(  p^{R_{0}}\left\vert x\right\vert _{p}\right)  \Psi_{\Gamma
}\left(  x,t\right)  =\\
p^{\frac{-R_{0}}{2}}\Omega\left(  p^{R_{0}}\left\vert x\right\vert
_{p}\right)  \left\{  e^{-iE_{gnd}t}+\left(  p-1\right)  \text{ }%
{\displaystyle\sum\limits_{m=0}^{R_{0}-1}}
p^{m}e^{-iE_{m}t}\right\}  .
\end{gather*}
Therefore%
\begin{equation}
\pi_{\Gamma,\Gamma}\left(  t\right)  =p^{-2R_{0}}\left\vert e^{-iE_{gnd}%
t}+\left(  p-1\right)  \text{ }%
{\displaystyle\sum\limits_{m=0}^{R_{0}-1}}
p^{m}e^{-iE_{m}t}\right\vert ^{2}.\label{Transition_gamma}%
\end{equation}

\subsubsection{Computation of $\pi_{r,\Gamma}\left(  t\right)  $}

We use that%
\begin{align*}
\pi_{r,\Gamma}\left(  t\right)   &  =%
{\displaystyle\int\limits_{S_{-r}}}
\left\vert 1_{s_{-r}}(x)\Psi_{\Gamma}\left(  x,t\right)  \right\vert ^{2}dx\\
&  =p^{-r}%
{\displaystyle\int\limits_{S_{0}}}
\left\vert 1_{s_{-r}}(p^{r}u)\Psi_{\Gamma}\left(  p^{r}u,t\right)  \right\vert
^{2}du,
\end{align*}
where
\[
1_{s_{-r}}(x)\Psi_{\Gamma}\left(  x,t\right)  =p^{\frac{R_{0}}{2}}\left\{
C_{0}^{\prime}e^{-iE_{gnd}t}1_{s_{-r}}(x)+%
{\displaystyle\sum\limits_{k\in\mathbb{F}_{p}\smallsetminus\left\{  0\right\}
}}
\text{ }%
{\displaystyle\sum\limits_{m=0}^{R_{0}-1}}
C_{\left(  -m\right)  0k}^{\prime}e^{-iE_{m}t}1_{s_{-r}}(x)\Psi_{\left(
-m\right)  0k}\left(  x\right)  \right\}  ,
\]
and by (\ref{Restriction_2}), and Lemma \ref{Lemma7},
\[
1_{s_{-r}}(x)\Psi_{\Gamma}\left(  x,t\right)  =p^{\frac{-R_{0}}{2}}1_{s_{-r}%
}(x)\left\{  e^{-iE_{gnd}t}+%
{\displaystyle\sum\limits_{k\in\mathbb{F}_{p}\smallsetminus\left\{  0\right\}
}}
\text{ }%
{\displaystyle\sum\limits_{m=0}^{\min\left\{  r,R_{0}-1\right\}  }}
p^{m}e^{-iE_{m}t}\chi_{p}(p^{-m-1}kx)\right\}  .
\]
Therefore,%
\begin{equation}
\pi_{r,\Gamma}\left(  t\right)  =p^{-r-\frac{R_{0}}{2}}%
{\displaystyle\int\limits_{S_{0}}}
\left\vert e^{-iE_{gnd}t}+\left(  p-1\right)  \text{ }%
{\displaystyle\sum\limits_{m=0}^{\min\left\{  r,R_{0}-1\right\}  }}
p^{m}e^{-iE_{m}t}\chi_{p}(p^{-m-1+r}u)\right\vert ^{2}%
du.\label{Transition_gamma_J}%
\end{equation}

\subsubsection{Computation of $\pi_{\Gamma,r}\left(  t\right)  $}

We use that%
\[
\pi_{\Gamma,r}\left(  t\right)  =%
{\displaystyle\int\limits_{p^{R_{0}}\mathbb{Z}_{p}}}
\left\vert \Omega\left(  p^{R_{0}}\left\vert x\right\vert _{p}\right)
\Psi_{r}\left(  x,t\right)  \right\vert ^{2}dx,
\]
where%
\begin{align*}
\Psi_{r}\left(  x,t\right)   &  =p^{\frac{r}{2}}\left(  1-p^{-1}\right)
^{\frac{-1}{2}}e^{-i\boldsymbol{H}t}1_{S_{-r}}(x)\\
&  =p^{\frac{r}{2}}\left(  1-p^{-1}\right)  ^{\frac{-1}{2}}\left\{
C_{0}e^{-iE_{gnd}t}\Omega\left(  \left\vert x\right\vert _{p}\right)  +%
{\displaystyle\sum\limits_{k\in\mathbb{F}_{p}\smallsetminus\left\{  0\right\}
}}
\text{ }%
{\displaystyle\sum\limits_{m=0}^{r}}
C_{\left(  -m\right)  0k}e^{-iE_{m}t}\Psi_{\left(  -m\right)  0k}\left(
x\right)  \right\}  ,
\end{align*}
and by (\ref{Restriction}),
\begin{align*}
\Omega\left(  p^{R_{0}}\left\vert x\right\vert _{p}\right)  \Psi_{r}\left(
x,t\right)   &  =p^{\frac{r}{2}}\left(  1-p^{-1}\right)  ^{\frac{-1}{2}}%
\Omega\left(  p^{R_{0}}\left\vert x\right\vert _{p}\right)  \times\\
&  \left\{  C_{0}e^{-iE_{gnd}t}+%
{\displaystyle\sum\limits_{k\in\mathbb{F}_{p}\smallsetminus\left\{  0\right\}
}}
{\displaystyle\sum\limits_{m=0}^{\min\left\{  r,R_{0}-1\right\}  }}
C_{\left(  -m\right)  0k}e^{-iE_{m}t}p^{\frac{m}{2}}\right\}  .
\end{align*}
Therefore,%
\begin{equation}
\pi_{\Gamma,r}\left(  t\right)  =p^{\frac{r}{2}-R_{0}}\left(  1-p^{-1}\right)
^{\frac{-1}{2}}\left\vert C_{0}e^{-iE_{gnd}t}+%
{\displaystyle\sum\limits_{k\in\mathbb{F}_{p}\smallsetminus\left\{  0\right\}
}}
{\displaystyle\sum\limits_{m=0}^{\min\left\{  r,R_{0}-1\right\}  }}
C_{\left(  -m\right)  0k}e^{-iE_{m}t}p^{\frac{m}{2}}\right\vert ^{2}%
,\label{Transition_R_gamma}%
\end{equation}
where the constants $C_{0},C_{\left(  -m\right)  0k}$ are given in the Lemma
\ref{Lemma6}.

\begin{theorem}
The transition probabilities $\pi_{\Gamma,\Gamma}\left(  t\right)  $%
,$\pi_{r,\Gamma}\left(  t\right)  ,\pi_{\Gamma,r}\left(  t\right)  $ are given
by the formulas (\ref{Transition_gamma}), (\ref{Transition_gamma_J}),
(\ref{Transition_R_gamma}).
\end{theorem}

\section{Discussion}

In the Dirac-von Neumann formulation of QM, the states of a quantum system are
vectors of an abstract complex Hilbert space $\mathcal{H}$, and the
observables correspond to linear self-adjoint operators in $\mathcal{H}$. A
particular choice of space $\mathcal{H}$ goes beyond the mathematical
formulation and belongs to the domain of physical practice and intuition. In
practice, choosing a particular Hilbert space also implies choosing a topology
for the space. For instance, if we take $\mathcal{H}=L%
{{}^2}%
(\mathbb{R}^{N})$, we are assuming that space ($\mathbb{R}^{N}$) is
continuous, which means that there is a continuous path joining any two points
in the space. The space $\mathbb{Q}_{p}^{N}$ is discrete, there is no a curve
joining two different points. In the $p$-adic QM, $\mathcal{H}=L^{2}%
(\mathbb{Q}_{p}^{N})$.

In the 1930s Bronstein showed that general relativity and quantum mechanics
imply that the uncertainty $\Delta x$ of any length measurement satisfies
\[
\Delta x\geq L_{\text{Planck}}:=\sqrt{\frac{\hbar G}{c^{3}}},
\]
where $L_{\text{Planck}}$ is the Planck length ($L_{\text{Planck}}%
\approx10^{-33}$ $cm$). By interpreting this inequality as the nonexistence of
`intervals' below the Planck scale, one is driven naturally to conclude that
the physical space at a very short distance is a completely disconnected
topological space; intuitively, the space is just a collection of isolated
points. Take $\mathbb{Q}_{p}^{N}$ as a model of the physical space implies
that the Planck length is $p^{-1}$, i.e., the smallest distance between two
different points, up to scale transformations, is $p^{-1}$; see
\cite{Zuniga-AP} for a further discussion. On the other hand, if we replace
$\mathbb{Q}_{p}^{N}$ by $\mathbb{R}^{N}$, there is no Planck length due to the
Archimedean axiom of $\mathbb{R}$.

In \cite{Zuniga-AP}, a $p$-adic model of the double-slit experiment was
studied; in this model, each particle goes through one slit only. A similar
description of the two-slit experiment was given in \cite{Aharonov et al}:
"Instead of a quantum wave passing through both slits, we have a localized
particle with nonlocal interactions with the other slit." in \cite{Zuniga-AP},
the same conclusion was obtained, but in the $p$-adic framework, the nonlocal
interactions are a consequence of the discreteness of the space $\mathbb{Q}%
_{p}^{N}$.

On the other hand, in \cite{Zuniga-PhA} was studied the breaking of the
Lorentz symmetry at the Planck length in quantum mechanics. This work used
three-dimensional $p$-adic vectors as position variables while the time
remains a real number. The Lorentz symmetry is naturally broken. A new
$p$-adic Dirac equation was introduced; it predicts the existence of particles
and antiparticles and charge conjugation like the standard one. The
discreteness of the $p$-adic space imposes substantial restrictions on the
solutions of the new equation. This equation admits localized solutions, which
is impossible in the standard case. It was shown that an isolated quantum
system whose evolution is controlled by the $p$-adic Dirac equation does not
satisfy the Einstein causality, which means that the speed of light is not the
upper limit for the speed at which conventional matter or energy can travel
through space. The new $p$-adic Dirac equation is not intended to replace the
standard one; it should be understood as a new version (or a limit) of the
classical equation at the Planck length scale.

Besides the promising results mentioned above, a drawback of the $p$-adic QM
is that it requires energy regimes so high as to create a black hole. In this
article, we introduce a $p$-adic quantum model with a clear physical meaning
that does require the extremely high energy regimes mentioned. In this
article, we study the solutions of the $p$-adic Schr\"{o}dinger equations for
infinite well potentials:
\begin{equation}
\left\{
\begin{array}
[c]{l}%
i\hbar\frac{\partial\Psi(x,t)}{\partial t}=\left(  m_{\alpha}\boldsymbol{D}%
^{\alpha}+V(\left\Vert x\right\Vert _{p})\right)  \Psi(x,t)\text{, }%
x\in\mathbb{Z}_{p}^{N},\quad t\geq0\\
\\
\text{\textrm{supp}}\Psi(x,t)\subset\mathbb{Z}_{p}^{N}\text{, }t\geq0\\
\\
\Psi(x,0)=\Psi_{0}(x)\in L^{2}\left(  \mathbb{Z}_{p}^{N}\right)  ,
\end{array}
\right.  \label{Schrodinger_Equation_10}%
\end{equation}
where $m_{\alpha}>0$, $\boldsymbol{D}^{\alpha}$, $\alpha>0$, is the
Taibleson-Vladimirov operator, and%
\[
V(x)=V(\left\Vert x\right\Vert _{p})=\left\{
\begin{array}
[c]{ccc}%
0 & \text{if} & x\in\mathbb{Z}_{p}^{N}\\
&  & \\
\infty & \text{if} & x\notin\mathbb{Z}_{p}^{N},
\end{array}
\right.
\]
is a infinite potential well supported in the $N$-dimensional $p$-adic ball.
We solve rigorously problem (\ref{Schrodinger_Equation_10}). By using the
partition%
\begin{equation}
\mathbb{Z}_{p}^{N}\setminus\left\{  0\right\}  =%
{\displaystyle\bigsqcup\limits_{j=0}^{\infty}}
S_{-j}^{N}, \label{partition}%
\end{equation}
we construct a CTQW, where a walker jumps between the spheres $S_{-v}^{N}$ to
$S_{-r}^{N}$ with a transition probability $\pi_{r,v}\left(  t\right)  $,
which is computed using the Born rule as follows. Take $\Psi(x,t)$ the
wavefunction\ satisfying $\Psi(x,0)=p^{\frac{v}{2}}\left(  1-p^{-N}\right)
^{\frac{-1}{2}}1_{S_{-v}^{N}}(x)$ , then%
\[
\pi_{r,v}\left(  t\right)  =%
{\displaystyle\int\limits_{S_{-r}^{N}}}
\left\vert \Psi(x,t)\right\vert ^{2}d^{N}x.
\]
This construction can be extended to other partitions of the unit ball for
which the calculation of the $\pi_{r,v}\left(  t\right)  $ is amenable. Here,
it is relevant to emphasize that the construction of the mentioned CTQWs
depends on the fractal nature of the $N$-dimensional $p$-adic ball. In
particular, in $\mathbb{R}^{N}$, there is no counterpart of partition
(\ref{partition}).The introduced CTQWs cannot be constructed using the
Farhi-Gutmann approach. This technique is based on adjacency matrices, which
do not play a significant role in our construction.

In our view, the result is highly relevant for two reasons. First, it
establishes a connection between $p$-adic quantum mechanics and quantum
computing. A relevant research question is to investigate if $p$-adic QM can
be used as a toolbox to construct and analyze CTQWs. Second, the model
presented suggests that the space is discrete. As we mentioned before, it is
widely accepted that the space at the level of the Planck length is discrete;
here, our results suggest that the condition on the Planck length is not required.

\section{\label{Appendix} Appendix: Basic facts on $p$-adic analysis}

In this section, we fix the notation and collect some basic results on
$p$-adic analysis that we will use throughout the article. For a detailed
exposition on $p$-adic analysis, the reader may consult \cite{V-V-Z},
\cite{Alberio et al}, \cite{Taibleson}, \cite{Zuniga-Textbook}.

\subsection{The field of $p$-adic numbers}

Along this article $p$ denotes a prime number. The field of $p-$adic numbers
$\mathbb{Q}_{p}$ is defined as the completion of the field of rational numbers
$\mathbb{Q}$ with respect to the $p-$adic norm $|\cdot|_{p}$, which is defined
as
\[
|x|_{p}=%
\begin{cases}
0 & \text{if }x=0\\
p^{-\gamma} & \text{if }x=p^{\gamma}\dfrac{a}{b},
\end{cases}
\]
where $a$ and $b$ are integers coprime with $p$. The integer $\gamma
=ord_{p}(x):=ord(x)$, with $ord(0):=+\infty$, is called the\textit{\ }$p-$adic
order of $x$. We extend the $p-$adic norm to $\mathbb{Q}_{p}^{N}$ by taking%
\[
||x||_{p}:=\max_{1\leq i\leq N}|x_{i}|_{p},\qquad\text{for }x=(x_{1}%
,\dots,x_{N})\in\mathbb{Q}_{p}^{N}.
\]
By defining $ord(x)=\min_{1\leq i\leq N}\{ord(x_{i})\}$, we have
$||x||_{p}=p^{-ord(x)}$.\ The metric space $\left(  \mathbb{Q}_{p}^{N}%
,||\cdot||_{p}\right)  $ is a complete ultrametric space. As a topological
space $\mathbb{Q}_{p}$\ is homeomorphic to a Cantor-like subset of the real
line, see, e.g., \cite{V-V-Z}, \cite{Alberio et al}.

Any $p-$adic number $x\neq0$ has a unique expansion of the form
\[
x=p^{ord(x)}\sum_{j=0}^{\infty}x_{j}p^{j},
\]
where $x_{j}\in\{0,1,2,\dots,p-1\}$ and $x_{0}\neq0$. \ In addition, any
$x\in\mathbb{Q}_{p}^{N}\smallsetminus\left\{  0\right\}  $ can be represented
uniquely as $x=p^{ord(x)}v$, where $\left\Vert v\right\Vert _{p}=1$.

\subsection{Topology of $\mathbb{Q}_{p}^{N}$}

For $r\in\mathbb{Z}$, denote by $B_{r}^{N}(a)=\{x\in\mathbb{Q}_{p}%
^{N};||x-a||_{p}\leq p^{r}\}$ the ball of radius $p^{r}$ with center at
$a=(a_{1},\dots,a_{N})\in\mathbb{Q}_{p}^{N}$, and take $B_{r}^{N}%
(0):=B_{r}^{N}$. Note that $B_{r}^{N}(a)=B_{r}(a_{1})\times\cdots\times
B_{r}(a_{N})$, where $B_{r}(a_{i}):=\{x\in\mathbb{Q}_{p};|x_{i}-a_{i}|_{p}\leq
p^{r}\}$ is the one-dimensional ball of radius $p^{r}$ with center at
$a_{i}\in\mathbb{Q}_{p}$. The ball $B_{0}^{N}$ equals the product of $N$
copies of $B_{0}=\mathbb{Z}_{p}$, the ring of $p-$adic integers. We also
denote by $S_{r}^{N}(a)=\{x\in\mathbb{Q}_{p}^{N};||x-a||_{p}=p^{r}\}$ the
sphere of radius $p^{r}$ with center at $a=(a_{1},\dots,a_{N})\in
\mathbb{Q}_{p}^{N}$, and take $S_{r}^{N}(0):=S_{r}^{N}$. We notice that
$S_{0}^{1}=\mathbb{Z}_{p}^{\times}$ (the group of units of $\mathbb{Z}_{p}$),
but $\left(  \mathbb{Z}_{p}^{\times}\right)  ^{N}\subsetneq S_{0}^{N}$. The
balls and spheres are both open and closed subsets in $\mathbb{Q}_{p}^{N}$. In
addition, two balls in $\mathbb{Q}_{p}^{N}$ are either disjoint or one is
contained in the other.

As a topological space $\left(  \mathbb{Q}_{p}^{N},||\cdot||_{p}\right)  $ is
totally disconnected, i.e., the only connected \ subsets of $\mathbb{Q}%
_{p}^{N}$ are the empty set and the points. A subset of $\mathbb{Q}_{p}^{N}$
is compact if and only if it is closed and bounded in $\mathbb{Q}_{p}^{N}$,
see, e.g., \cite[Section 1.3]{V-V-Z}, or \cite[Section 1.8]{Alberio et al}.
The balls and spheres are compact subsets. Thus $\left(  \mathbb{Q}_{p}%
^{N},||\cdot||_{p}\right)  $ is a locally compact topological space.

\subsection{The Haar measure}

Since $(\mathbb{Q}_{p}^{N},+)$ is a locally compact topological group, there
exists a Haar measure $d^{N}x$, which is invariant under translations, i.e.,
$d^{N}(x+a)=d^{N}x$, \cite{Halmos}. If we normalize this measure by the
condition $\int_{\mathbb{Z}_{p}^{N}}dx=1$, then $d^{N}x$ is unique.

\begin{notation}
We will use $\Omega\left(  p^{-r}||x-a||_{p}\right)  $ to denote the
characteristic function of the ball $B_{r}^{N}(a)=a+p^{-r}\mathbb{Z}_{p}^{N}$,
where
\[
\mathbb{Z}_{p}^{N}=\left\{  x\in\mathbb{Q}_{p}^{N};\left\Vert x\right\Vert
_{p}\leq1\right\}
\]
is the $N$-dimensional unit ball. For more general sets, we will use the
notation $1_{A}$ for the characteristic function of set $A$.
\end{notation}

\subsection{The Bruhat-Schwartz space}

A complex-valued function $\varphi$ defined on $\mathbb{Q}_{p}^{N}$ is called
locally constant if for any $x\in\mathbb{Q}_{p}^{N}$ there exist an integer
$l(x)\in\mathbb{Z}$ such that%
\begin{equation}
\varphi(x+x^{\prime})=\varphi(x)\text{ for any }x^{\prime}\in B_{l(x)}^{N}.
\label{local_constancy}%
\end{equation}
A function $\varphi:\mathbb{Q}_{p}^{N}\rightarrow\mathbb{C}$ is called a
Bruhat-Schwartz function (or a test function) if it is locally constant with
compact support. Any test function can be represented as a linear combination,
with complex coefficients, of characteristic functions of balls. The
$\mathbb{C}$-vector space of Bruhat-Schwartz functions is denoted by
$\mathcal{D}(\mathbb{Q}_{p}^{N})$. For $\varphi\in\mathcal{D}(\mathbb{Q}%
_{p}^{N})$, the largest number $l=l(\varphi)$ satisfying
(\ref{local_constancy}) is called the exponent of local constancy (or the
parameter of constancy) of $\varphi$.

\subsection{$L^{\rho}$ spaces}

Given $\rho\in\lbrack1,\infty)$, we denote by$L^{\rho}\left(
\mathbb{Q}
_{p}^{N}\right)  :=L^{\rho}\left(
\mathbb{Q}
_{p}^{N},d^{N}x\right)  ,$ the $\mathbb{C}-$vector space of all the complex
valued functions $g$ satisfying
\[
\left\Vert g\right\Vert _{\rho}=\left(  \text{ }%
{\displaystyle\int\limits_{\mathbb{Q}_{p}^{N}}}
\left\vert g\left(  x\right)  \right\vert ^{\rho}d^{N}x\right)  ^{\frac
{1}{\rho}}<\infty,
\]
where $d^{N}x$ is the normalized Haar measure on $\left(  \mathbb{Q}_{p}%
^{N},+\right)  $.

If $U$ is an open subset of $\mathbb{Q}_{p}^{N}$, $\mathcal{D}(U)$ denotes the
$\mathbb{C}$-vector space of test functions with supports contained in $U$,
then $\mathcal{D}(U)$ is dense in
\[
L^{\rho}\left(  U\right)  =\left\{  \varphi:U\rightarrow\mathbb{C};\left\Vert
\varphi\right\Vert _{\rho}=\left\{
{\displaystyle\int\limits_{U}}
\left\vert \varphi\left(  x\right)  \right\vert ^{\rho}d^{N}x\right\}
^{\frac{1}{\rho}}<\infty\right\}  ,
\]
for $1\leq\rho<\infty$, see, e.g., \cite[Section 4.3]{Alberio et al}. We
denote by $L_{\mathbb{R}}^{\rho}\left(  U\right)  $ the real counterpart of
$L^{\rho}\left(  U\right)  $. We use mainly the case where $U=p^{L}%
\mathbb{Z}_{p}^{N}$.

\subsection{The Fourier transform}

Set $\chi_{p}(y)=\exp(2\pi i\{y\}_{p})$ for $y\in\mathbb{Q}_{p}$, and
$\xi\cdot x:=\sum_{j=1}^{N}\xi_{j}x_{j}$, for $\xi=(\xi_{1},\dots,\xi_{N})$,
$x=(x_{1},\dots,x_{N})\allowbreak\in\mathbb{Q}_{p}^{N}$, as before. The
Fourier transform of $\varphi\in\mathcal{D}(\mathbb{Q}_{p}^{N})$ is defined
as
\[
\mathcal{F}\varphi(\xi)=%
{\displaystyle\int\limits_{\mathbb{Q}_{p}^{N}}}
\chi_{p}(\xi\cdot x)\varphi(x)d^{N}x\quad\text{for }\xi\in\mathbb{Q}_{p}^{N},
\]
where $d^{N}x$ is the normalized Haar measure on $\mathbb{Q}_{p}^{N}$. The
Fourier transform is a linear isomorphism from $\mathcal{D}(\mathbb{Q}_{p}%
^{N})$ onto itself satisfying
\begin{equation}
(\mathcal{F}(\mathcal{F}\varphi))(\xi)=\varphi(-\xi), \label{Eq_FFT}%
\end{equation}
see, e.g., \cite[Section 4.8]{Alberio et al}. We also use the notation
$\mathcal{F}_{x\rightarrow\kappa}\varphi$ and $\widehat{\varphi}$\ for the
Fourier transform of $\varphi$.

The Fourier transform extends to $L^{2}\left(  \mathbb{Q}_{p}^{N}\right)  $.
If $f\in L^{2}\left(  \mathbb{Q}_{p}^{N}\right)  $, its Fourier transform is
defined as
\[
(\mathcal{F}f)(\xi)=\lim_{k\rightarrow\infty}%
{\displaystyle\int\limits_{||x||_{p}\leq p^{k}}}
\chi_{p}(\xi\cdot x)f(x)d^{N}x,\quad\text{for }\xi\in%
\mathbb{Q}
_{p}^{N},
\]
where the limit is taken in $L^{2}\left(  \mathbb{Q}_{p}^{N}\right)  $. We
recall that the Fourier transform is unitary on $L^{2}\left(  \mathbb{Q}%
_{p}^{N}\right)  $, i.e. $||f||_{2}=||\mathcal{F}f||_{2}$ for $f\in
L^{2}\left(  \mathbb{Q}_{p}^{N}\right)  $ and that (\ref{Eq_FFT}) is also
valid in $L^{2}\left(  \mathbb{Q}_{p}^{N}\right)  $, see, e.g., \cite[Chapter
III, Section 2]{Taibleson}.

\subsection{The Taibleson-Vladimirov operator}

We denote by $\mathcal{C}(\mathbb{Q}_{p}^{N})$ the $\mathbb{C}$-vector space
of continuous functions defined on $\mathbb{Q}_{p}^{N}$. The
Taibleson-Vladimirov pseudo-differential operator $\boldsymbol{D}^{\alpha}$,
$\alpha>0$, is defined as%
\begin{equation}%
\begin{array}
[c]{cccc}%
\boldsymbol{D}^{\alpha}: & \mathcal{D}(\mathbb{Q}_{p}^{N}) & \rightarrow &
L^{2}(\mathbb{Q}_{p}^{N})\cap\mathcal{C}(\mathbb{Q}_{p}^{N})\\
&  &  & \\
& \varphi(x) & \rightarrow & \boldsymbol{D}^{\alpha}\varphi(x)=\mathcal{F}%
_{\xi\rightarrow x}^{-1}\left(  ||\xi||_{p}^{\alpha}\mathcal{F}_{x\rightarrow
\xi}\left(  \varphi\right)  \right)  .
\end{array}
\label{Taibleson_Vladimirov_Operator}%
\end{equation}
This operator admits an extension of the form
\[
\left(  \boldsymbol{D}^{\alpha}\varphi\right)  \left(  x\right)
=\frac{1-p^{\alpha}}{1-p^{-\alpha-N}}\int\limits_{\mathbb{Q}_{p}^{N}}%
||y||_{p}^{-\alpha-N}(\varphi(x-y)-\varphi(x))\,d^{N}y
\]
to the space of locally locally constant functions $\varphi(x)$ satisfying
\[
\int\limits_{||x||_{p}\geq1}||x||_{p}^{-\alpha-N}|\varphi(x)|\,d^{N}x<\infty,
\]
see, e.g., \cite[Chapter 2]{Zuniga-Textbook}. In particular, the
Taibleson-Vladimirov derivative of any order of a constant function is zero.

\end{document}